\documentclass[11pt]{article}

\usepackage{amsmath, amssymb}

\usepackage{amsthm}

\usepackage{ifthen}

\usepackage{ifpdf}

\ifpdf
\usepackage[pdftex]{graphicx}
\usepackage[pdftex]{hyperref}
\else
\usepackage[dvips]{graphicx}
\fi

\ifthenelse{\isundefined{\NoGenerateLetterSize}}
{
\ifpdf
\setlength{\pdfpagewidth}{8.5in}
\setlength{\pdfpageheight}{11in}
\else

\fi
}
{}

\newcommand{\Comment}[1]{\ignorespaces}

\long\def\LongVersion#1\LongVersionEnd{#1}
\long\def\ShortVersion#1\ShortVersionEnd{}

\usepackage{tikz}
\usetikzlibrary{automata}
\usetikzlibrary[shapes.multipart]
\usetikzlibrary{arrows}

\LongVersion %{
\LongVersionEnd %}

\ShortVersion %{
\ShortVersionEnd %}

\ShortVersion %{
\usepackage{times}
\ShortVersionEnd %}

\ShortVersion %{
\renewcommand{\paragraph}[1]{\par\noindent\textbf{#1}}
\ShortVersionEnd %}

\newtheorem{theorem}{Theorem}[section]
\newtheorem{lemma}[theorem]{Lemma}
\newtheorem{observation}[theorem]{Observation}
\newtheorem{corollary}[theorem]{Corollary}
\newtheorem{proposition}[theorem]{Proposition}

\newtheorem*{theorem*}{Theorem}
\newtheorem*{observation*}{Observation}

\theoremstyle{definition}

\newtheorem*{remark*}{Remark}

\theoremstyle{plain}
\newenvironment{AvoidOverfullParagraph}[0]
{\sloppy\ignorespaces}
{\par\fussy\ignorespacesafterend}

\newenvironment{DenseItemize}[0]
{\begin{itemize} \itemsep0pt \parskip0pt \parsep0pt}
{\end{itemize}}

\newcommand{\Integers}[0]{\mathbb{Z}}
\newcommand{\Neighbors}[0]{\mathit{N}}
\newcommand{\Expectation}[0]{\mathbb{E}}
\newcommand{\Probability}[0]{\mathbb{P}}
\newcommand{\Number}[0]{\sharp}
\newcommand{\BoundPar}[0]{{^{\geq}b}}

\newcommand{\Win}[0]{\mathtt{WIN}}
\newcommand{\Lose}[0]{\mathtt{LOSE}}
\newcommand{\Up}[0]{\mathtt{UP}}
\newcommand{\Down}[0]{\mathtt{DOWN}}
\newcommand{\WinShort}[0]{\mathtt{W}}
\newcommand{\LoseShort}[0]{\mathtt{L}}
\newcommand{\UpShort}[0]{\mathtt{U}}
\newcommand{\DownShort}[0]{\mathtt{D}}
\newcommand{\Geom}[0]{\mathrm{Geom}}
\newcommand{\NegativeBinomial}[0]{\mathrm{NB}}

\newcommand{\Colored}[0]{\mathtt{COLORED}}
\newcommand{\Active}[0]{\mathtt{ACTIVE}}
\newcommand{\Waiting}[0]{\mathtt{WAITING}}
\newcommand{\RandomColoring}[0]{\mathtt{RandColor}}

\LongVersion %{
\newcommand{\Section}[0]{Section}
\newcommand{\Appendix}[0]{Appendix}
\newcommand{\Figure}[0]{Figure}
\LongVersionEnd %}
\ShortVersion %{
\newcommand{\Section}[0]{Sec.}
\newcommand{\Appendix}[0]{App.}
\newcommand{\Figure}[0]{Fig.}
\ShortVersionEnd %}

\begin{document}

\title{Stone Age Distributed
\LongVersion %{
Computing
\LongVersionEnd %}
\ShortVersion %{
Computing\footnote{
A preliminary version of this paper was uploaded to arXiv (as suggested in the
FOCS CFP).
% Shortly afterwards, the arXiv version got a bit of press coverage, for
% instance by Technology Review \cite{TRonStoneAge}, but also by other tech and
% science media.
Shortly afterwards, the arXiv version got a bit of press coverage (Technology
Review and some other science and tech media) which was restricted to the
possible applications of the paper and in any case, did not include any
discussion of the techniques and proofs. 
To the best of our understanding, this does not contradict the prior
publication and simultaneous submission guidelines of FOCS.}
\ShortVersionEnd %}
}

\author{
Yuval Emek
\and
Jasmin Smula
\and
Roger Wattenhofer
\and \\
Computer Engineering and Networks Laboratory (TIK) \\
ETH Zurich, Switzerland
}
\date{}

\begin{titlepage}

\maketitle

\begin{abstract}
The traditional models of distributed computing focus mainly on networks of
computer-like devices that can exchange large messages with their neighbors
and perform arbitrary local computations.
Recently, there is a trend to apply distributed computing methods to networks
of sub-microprocessor devices, e.g., biological cellular networks or networks
of nano-devices.
However, the suitability of the traditional distributed computing models to
these types of networks is questionable:
do tiny bio/nano nodes ``compute'' and/or  ``communicate'' essentially the
same as a computer?
In this paper, we introduce a new model that depicts a network of randomized
finite state machines operating in an asynchronous environment.
Although the computation and communication capabilities of each individual
device in the new model are, by design, much weaker than those of a computer,
we show that some of the most important and extensively studied distributed
computing problems can still be solved efficiently.
\end{abstract}

\renewcommand{\thepage}{}
\end{titlepage}

\pagenumbering{arabic}

%%%%%%%%%%%%%%%%%%%%%%%%%%%%%%%%%%%%%%%%%%%%%%%%%%%%%%%%%%%%%%%%%%%%%%%%%%%%%%
\section{Introduction}
%%%%%%%%%%%%%%%%%%%%%%%%%%%%%%%%%%%%%%%%%%%%%%%%%%%%%%%%%%%%%%%%%%%%%%%%%%%%%%
Networks are at the core of many scientific areas, be it social sciences
(where networks for instance model human relations), logistics (e.g. traffic),
or electrical engineering (e.g. circuits).
\emph{Distributed computing} is the area that studies the power and
limitations of distributed algorithms and computation in networks.
Due to the major role that the Internet plays today, models targeted at
understanding the fundamental properties of networks focus mainly on
``Internet-capable'' devices.
The standard model in distributed computing is the so called \emph{message
passing} model, where nodes may exchange large messages with their neighbors,
and perform arbitrary local computations.

Some networks though, are not truthfully represented by the classical message
passing model.
For example, \emph{wireless} networks such as ad hoc or sensor networks, whose
research has blossomed in the last decade, require some adaptations of the
message passing model so that it meets the limited capabilities of the
underlying wireless devices more precisely.
More recently, there is a trend to apply distributed computing methods, and in
particular, the message passing model, to networks of sub-microprocessor
devices, for instance networks of biological cells or nano-scale mechanical
devices.
However, the suitability of the message passing model to these types of
networks is far from being certain:
do tiny bio/nano nodes ``compute'' and/or  ``communicate'' essentially the
same as a computer?
Since such nodes will be fundamentally more limited than silicon-based devices,
we believe that there is a need for a network model, where nodes
are by design below the computation and communication capabilities of Turing
machines.

%%%%%%%%%%%%%%%%%%%%%%%%%%%%%%%%%%%%%%%
\paragraph{Networked Finite State Machines.}
%%%%%%%%%%%%%%%%%%%%%%%%%%%%%%%%%%%%%%%
In this paper, we take a radically different approach:
Instead of imposing additional restrictions on the existing models for
networks of computer-like devices, we introduce an entirely new
model, referred to as \emph{networked finite state machines (nFSM)}, that
depicts a network of randomized finite state machines progressing
in asynchronous steps (refer to \Section{}~\ref{section:Model} for a formal
description).
Under the nFSM model, nodes communicate by transmitting messages belonging to
some finite communication alphabet $\Sigma$ such that a message $\sigma \in
\Sigma$ transmitted by node $u$ is delivered to its neighbors (the same
$\sigma$ to all neighbors) in an asynchronous fashion;
each neighbor $v$ of $u$ has a port corresponding to $u$ in which the last
message delivered from $u$ is stored.

The access of node $v$ to its ports is limited:
each state $q$ in the state set $Q$ of the FSM is associated with some
\emph{query letter} $\sigma = \sigma(q) \in \Sigma$;
if node $v$ resides in state $q$ at some step of the execution, then the next
state and the message transmitted by $v$ at this step are determined by $q$
and by the number $\Number(\sigma)$ of occurrences of $\sigma$ in $v$'s ports.
The crux of the model is that $\Number(\sigma)$ is calculated according to the
\emph{one-two-many}\footnote{
The one-two-many theory states that some small isolated cultures (e.g., the
Piraha tribe of the Amazon \cite{Gordon04}) did not develop a counting system
that goes beyond $2$.
This is reflected in their languages that include words for ``$1$'', ``$2$'',
and ``many'' that stands for any number larger than $2$.
} principle:
the node can only count up to some predetermined \emph{bounding parameter} $b
\in \Integers_{> 0}$ and any value of $\Number(\sigma)$ larger than $b$ cannot
be distinguished from $b$.

In particular, the nFSM model satisfies the following \emph{model
requirements}, that we believe, make it more applicable to the study of
networks consisting of weaker devices such as those mentioned above. \\
\textbf{(M1)}
The model is applicable to arbitrary network topologies. \\
\textbf{(M2)}
All nodes run the same protocol executed by a (randomized) FSM. \\
\textbf{(M3)}
The network operates within an asynchronous environment, with node activation
patterns independent of message delivery patterns. \\
\textbf{(M4)}
All features of the FSM (specifically, the state set $Q$, message alphabet
$\Sigma$, and bounding parameter $b$) are of constant size independent of any
parameter of the network (including the degree of the node executing the
FSM). \\
The last requirement is perhaps the most interesting one as it implies that a
node cannot perform any calculation that involves numbers beyond some
predetermined constant.
This comes in contrast to many distributed algorithms operating under the
message passing model that strongly rely on the ability of a node to perform
such calculations (e.g., count up to some parameter of the network or a
function thereof).

%%%%%%%%%%%%%%%%%%%%%%%%%%%%%%%%%%%%%%%
\paragraph{Results.}
%%%%%%%%%%%%%%%%%%%%%%%%%%%%%%%%%%%%%%%
Our investigation of the new model begins by implementing an nFSM synchronizer
that practically allows the algorithm designer to assume a synchronous
environment (\Section{}~\ref{section:ConvenientTransformations}).
Then, we show that the computational power of a network operating
under the nFSM model is essentially equivalent to that of a randomized Turing
machine with linear space bound (cf. linear bounded automaton).
In comparison, the computational power of a network operating
under the message passing model is trivially equivalent to that of a (general)
Turing machine, therefore there exist distributed problems that can be solved
under the message passing model in constant time but cannot be solved under
the nFSM model at all (\Section{}~\ref{section:LBA}).

Nevertheless, we show that arguably the most important and extensively studied
problems in distributed computing admit efficient --- namely, with run-time
polylogarithmic in the number of nodes --- algorithms operating under the nFSM
model.
Specifically, we develop such algorithms for computing a maximal independent
set (MIS) in arbitrary graphs (\Section{}~\ref{section:MIS}) and for
$3$-coloring of (undirected) trees
\LongVersion %{
(\Section{}~\ref{section:Coloring}).
\LongVersionEnd %}
\ShortVersion %{
(\Appendix{}~\ref{appendix:Coloring}).
\ShortVersionEnd %}
We also develop an efficient algorithm that computes a maximal matching in
arbitrary graphs, but this requires a small unavoidable modification of the
nFSM model that goes beyond the scope of the current version of the paper.

%%%%%%%%%%%%%%%%%%%%%%%%%%%%%%%%%%%%%%%
\paragraph{Related Work.}
%%%%%%%%%%%%%%%%%%%%%%%%%%%%%%%%%%%%%%%
As mentioned above, the message passing model is the gold standard when it
comes to understanding distributed algorithms.
Several variants exist for this model, differing mainly in the bounds imposed
on the message size and the level of synchronization.
Perhaps the most popular message passing variants are the fully synchronous
\emph{local} and \emph{congest} models \cite{Linial92,PelegBook,JukkaSurvey},
assuming that in each round, a node can send messages to its neighbors
(different messages to different neighbors), receive and interpret the
messages sent to it from its neighbors, and perform an arbitrary local
computation\footnote{
It is important to point out that even though the local and congest
models allow for arbitrary local computations, the existing literature hardly
ever assumes anything that cannot be computed in time polynomial in the size
of the information received thus far;
the rare exceptions are typically clearly mentioned in the text.
} determining, in particular, the messages sent in the next round.
The difference between the two variants is cast in the size of the
communicated messages:
the local model does not impose any restrictions on the message size, hence
it can be used for the purpose of establishing general lower bounds, whereas
the congest model is more information-theoretic, with a (typically
logarithmic) bound on the message size.
Indeed, most theoretical literature dealing with distributed
algorithms relies on one of these two models.

As the congest model still allows for sending different messages to
different neighbors in each round, it was too powerful for many settings.
Instead, with the proliferation of wireless networks, new more restrictive
message passing models appeared such as the \emph{radio network} model
\cite{ChlamtacKuten85}.
In radio networks, nodes still operate in synchronous rounds, where in
each round a node may choose to transmit a message or stay silent.
A transmitted message is received by all neighbors in the
network if the neighbors do not experience interference by concurrently
transmitting nodes in their own neighborhood.
There are several variants, e.g. whether nodes have collision detection, or
not.

Since the radio network model is still too powerful for some wireless
settings, more restrictive models were suggested.
One such example is the \emph{beeping} model
\cite{Flury2010Slotted,CornejoKuhn10}, where in each round a node can
either beep or stay silent, and a silent node can only distinguish between the
case in which no node in its neighborhood beeps and the case in which at least
one node beeps.
Efficient algorithms and lower bounds for the MIS problem under the beeping
model were developed by Afek et
al.~\cite{AfekAlonBarad+11,AfekAlonBarJoseph+11}.
Note that the beeping model resembles our nFSM model in the sense that the
``beeping rule'' can be viewed as counting under the one-two-many principle
with bounding parameter $b = 1$.
However, it is much stronger in other perspectives:
(i) the beeping model assumes synchronous communication and does not seem to
have a natural asynchronous variant, thus it does not satisfy requirement
(M3); and
(ii) the local computation is performed by a Turing machine whose memory is
allowed to grow with the network (this is crucial for the algorithms of Afek
et al.~\cite{AfekAlonBarad+11,AfekAlonBarJoseph+11}), thus it does not satisfy
requirements (M2) and (M4).

Our nFSM model is a generalization of the extensively studied \emph{cellular
automaton} model \cite{vonNeumann66,Gardner70,Wolfram02} that captures a
network of FSMs, arranged in a grid topology (some other highly regular
topologies were also considered), where the transition of each node depends on
its current state and the states of its neighbors.
Still, the nFSM model differs from the cellular automaton model in many
aspects;
in particular, the latter model is not applicable for non-regular network
topologies, in contrast to requirement (M1), and to the most part,
it also does not support asynchronous environments (at least not as asynchrony
is grasped in the current paper), in contrast to requirement (M3).

Another model that resembles the nFSM model is that of \emph{communicating
automata} \cite{BrandZafiropulo83}.
This model also assumes that each node in the network operates a FSM in an
asynchronous manner, however the steps of the FSMs are message driven:
for each state $q$ of node $v$ and for each message $m$ that node $v$ may
receive from an adjacent node $u$ while residing in state $q$, the transition
function of $v$ should have an entry characterized by the $3$-tuple $(q, u,
m)$ that determines its next move.
As such, different nodes would typically operate different FSMs, hence the
model does not satisfy requirement (M2), and more importantly, the size of the
FSM operated by node $v$ inherently depends on the degree of $v$, hence it
does not satisfy requirement (M4).
Moreover, the node activation pattern is driven by the incoming messages, so
it also does not satisfy requirement (M3).

%%%%%%%%%%%%%%%%%%%%%%%%%%%%%%%%%%%%%%%
\paragraph{Applicability to Biological Cellular Networks.}
%%%%%%%%%%%%%%%%%%%%%%%%%%%%%%%%%%%%%%%
Regardless of the theoretical interest in implementing efficient algorithms
using weaker assumptions, we believe that our new model and results should be
appealing to anyone interested in understanding the computational aspects of
biological cellular networks.
A basic dogma in biology (see, e.g., \cite{BiologyTextBook}) states that all
cells communicate and that they do so by emitting special kinds of proteins
(e.g., cytokines and chemokines in the immune system) that can be recognized
by designated receptors, thus enabling neighboring cells to distinguish
between different concentration levels of these proteins, which, after a
signaling cascade, leads to different gene expression.

Translated to the language of the nFSM model,
the emitted proteins correspond to the letters of the communication alphabet,
where the actual emission corresponds to transmitting a letter, and
the ability of a cell to distinguish between different concentration levels of
these proteins corresponds to the manner in which the nodes in our model
interpret the content of their ports.
Using an FSM as the underlying computational model of each node seems to be
the right choice especially in the biological setting as demonstrated by
Benenson et al.~\cite{Benenson+01} who showed that essentially any FSM can be
implemented by enzymes found in cells' nuclei.
One may wonder if the specific problems studied in the current paper have any
relevance to biological cellular networks.
Indeed, Afek et al.~\cite{AfekAlonBarad+11} discovered that a biological
process that occurs during the development of the nervous system of a fly is
in fact equivalent to solving the MIS problem.

%%%%%%%%%%%%%%%%%%%%%%%%%%%%%%%%%%%%%%%%%%%%%%%%%%%%%%%%%%%%%%%%%%%%%%%%%%%%%%
\section{Model}
\label{section:Model}
%%%%%%%%%%%%%%%%%%%%%%%%%%%%%%%%%%%%%%%%%%%%%%%%%%%%%%%%%%%%%%%%%%%%%%%%%%%%%%
Throughout, we assume a network represented by a finite undirected
graph $G = (V, E)$.
Under the \emph{networked finite state machines (nFSM)} model, each node $v
\in V$ runs a protocol depicted by the $8$-tuple
$$
\Pi = \left\langle Q, Q_I, Q_O, \Sigma, \sigma_0, b, \lambda, \delta
\right\rangle,
$$
where
\begin{DenseItemize}

\item
$Q$ is a finite set of \emph{states};

\item
$Q_I \subseteq Q$ is the subset of \emph{input states};

\item
$Q_O \subseteq Q$ is the subset of \emph{output states};

\item
$\Sigma$ is a finite \emph{communication alphabet};

\item
$\sigma_0 \in \Sigma$ is the \emph{initial letter};

\item
$b \in \mathbb{Z}_{>0}$ is a \emph{bounding parameter};
let $B = \{0, 1, \dots , b - 1, \BoundPar\}$ be a set of $b + 1$
distinguishable symbols;

\item
$\lambda : Q \rightarrow \Sigma$ assigns a \emph{query letter} $\sigma \in
\Sigma$ to every state $q \in Q$; and

\item
$\delta : Q \times B \rightarrow 2^{Q \times (\Sigma \cup \{\varepsilon\})}$
is the \emph{transition function}.

\end{DenseItemize}
It is important to point out that protocol $\Pi$ is oblivious to the graph
$G$.
In fact, the number of states in $Q$, the size of the alphabet $\Sigma$,
and the bounding parameter $b$ are all assumed to be universal constants,
independent of any parameter of the graph $G$.
In particular, the protocol executed by node $v \in V$ does not depend on the
degree of $v$ in $G$.
We now turn to describe the semantics of the nFSM model.

%%%%%%%%%%%%%%%%%%%%%%%%%%%%%%%%%%%%%%%
\paragraph{Communication.}
%%%%%%%%%%%%%%%%%%%%%%%%%%%%%%%%%%%%%%%
Node $v$ communicates with its adjacent nodes in $G$ by \emph{transmitting}
messages.
A transmitted message consists of a single letter $\sigma \in \Sigma$ and it
is assumed that this letter is delivered to all neighbors $u$ of $v$.
Each neighbor $u$ has a \emph{port} $\psi_{u}(v)$ (a different port for every
adjacent node $v$) in which the last message $\sigma$ received from $v$ is
stored.
At the beginning of the execution, all ports store the initial letter
$\sigma_0$.
It will be convenient to consider the case in which $v$ does not transmit any
message (and hence does not affect the corresponding ports of the adjacent
nodes) as a transmission of the special \emph{empty symbol} $\varepsilon$.

%%%%%%%%%%%%%%%%%%%%%%%%%%%%%%%%%%%%%%%
\paragraph{Execution.}
%%%%%%%%%%%%%%%%%%%%%%%%%%%%%%%%%%%%%%%
The execution of node $v$ progresses in discrete \emph{steps} indexed by the
positive integers.
At each step $t \in \Integers_{> 0}$, $v$ resides in some state $q \in Q$.
Let $\lambda(q) = \sigma \in \Sigma$ be the query letter that $\lambda$
assigns to state $q$ and let $\Number(\sigma)$ be the number of occurrences of
$\sigma$ in $v$'s ports in step $t$.
Then, the pair $(q', \sigma')$ of state $q' \in Q$ in which $v$ resides in step
$t + 1$ and message $\sigma' \in \Sigma \cup \{\varepsilon\}$ transmitted by
$v$ in step $t$ (recall that $\varepsilon$ indicates that no message is
transmitted) is chosen \emph{uniformly at random} (and independently of all
other random choices) among the pairs in
$$
\delta \left( q, f_{b} \left( \Number(\sigma) \right) \right) \subseteq Q
\times (\Sigma \cup \{\varepsilon\}) \ ,
$$
where $f_{b} : \Integers_{\geq 0} \rightarrow B$ is defined as
$$
f_{b}(x) =
\left\{
\begin{array}{ll}
x & \text{ if } 0 \leq x \leq b - 1 \ ; \\
\BoundPar & \text{ otherwise} \ . 
\end{array}
\right.
$$
Informally, this can be thought of as if $v$ queries its ports for occurrences
of $\sigma$ and ``observes'' the exact value of $\Number(\sigma)$ as long as it
is smaller than the bounding parameter $b$;
otherwise, $v$ merely ``observes'' that $\Number(\sigma) \geq b$ which is
indicated by the symbol $\BoundPar$.

%%%%%%%%%%%%%%%%%%%%%%%%%%%%%%%%%%%%%%%
\paragraph{Input and Output.}
%%%%%%%%%%%%%%%%%%%%%%%%%%%%%%%%%%%%%%%
Initially (in step $1$), each node resides in some of the input states in
$Q_{I}$.
The choice of the initial state of node $v \in V$ reflects the input passed to
$v$ at the beginning of the execution.
This allows our model to cope with distributed problems in which different
nodes get different input symbols.
When dealing with problems in which the nodes do not get any initial input
(such as the graph theoretic problems addressed in this paper), we shall
assume that $Q_{I}$ contains a single \emph{initial} state.

We say that the (global) execution of the protocol is in an \emph{output
configuration} if all nodes reside in output states of $Q_{O}$.
If this is the case, then the output of node $v \in V$ is determined by the
output state $q \in Q_{O}$ in which $v$ resides.

%%%%%%%%%%%%%%%%%%%%%%%%%%%%%%%%%%%%%%%
\paragraph{Asynchrony.}
%%%%%%%%%%%%%%%%%%%%%%%%%%%%%%%%%%%%%%%
The nodes are assumed to operate in an \emph{asynchronous} environment.
This asynchrony has two facets:
First, for the sake of convenience, we assume that the actual application of
the transition function in each step $t \in \Integers_{> 0}$ of node $v \in V$
is instantaneous (namely, lasts zero time) and occurs at the end of the
step;\footnote{
This assumption can be lifted at the cost of a more complicated definition of
the adversarial policy described soon.
}
the length of step $t$ of node $v$, denoted $L_{v, t}$, is defined as the time
difference between the application of the transition function in step $t - 1$
and that of step $t$.
It is assumed that $L_{v, t}$ is finite, but apart from that, we do not make
any further assumptions on this length, that is, the step length $L_{v, t}$ is
determined by the adversary independently of all other step lengths $L_{v', t'}$.
In particular, we do not assume any synchronization between the steps of
different nodes whatsoever.

Another facet of the asynchronous environment is that a message transmitted by
node $v$ in step $t$ (if such a message is transmitted) is assumed to reach
the port $\psi_{u}(v)$ of an adjacent node $u$ after a finite time delay,
denoted $D_{v, t, u}$.
We assume that if $v$ transmits message $\sigma_1 \in \Sigma$ in step $t_1$
and message $\sigma_2 \in \Sigma$ in step $t_2 > t_1$, then $\sigma_1$ reaches
$u$ before $\sigma_2$ does.
Apart from this ``FIFO'' assumption, we do not make any other assumptions on
the delays $D_{v, t, u}$.
In particular, this means that under certain circumstances, the
adversary may overwrite message $\sigma_1$ with message $\sigma_2$ in port
$\psi_{u}(v)$ of $u$ so that $u$ will never ``know'' that message $\sigma_1$
was transmitted.\footnote{
Often, much stronger assumptions are made in the literature.
For example, a common assumption for asynchronous environments is that the
port of node $u$ corresponding to the adjacent node $v$ is implemented by a
buffer so that messages cannot be ``lost''.
We do not make any such assumption for our nFSM model.}

Consequently, a \emph{policy} of the adversary is captured by:
(1) the length $L_{v, t}$ of step $t$ of node $v$ for every $v \in V$ and $t
\in \Integers_{> 0}$; and
(2) the delay $D_{v, t, u}$ of the delivery of the transmission of node $v$
in step $t$ to an adjacent node $u$ for every $v \in V$, $t \in
\Integers_{> 0}$, and $u \in \Neighbors(v)$.\footnote{
We use the standard notation $\Neighbors(v)$ for the \emph{neighborhood} of
node $v$ in $G$, namely, the subset of nodes adjacent to $v$.
}
Assuming that the adversary is oblivious to the random coin tosses of the
nodes, an adversarial policy is depicted by infinite sequences of $L_{v, t}$
and $D_{v, t, u}$ parameters.

For further information on asynchronous environments, we point the reader to
one of the standard textbooks \cite{PelegBook,LynchBook}.

%%%%%%%%%%%%%%%%%%%%%%%%%%%%%%%%%%%%%%%
\paragraph{Correctness and Run-Time Measures.}
%%%%%%%%%%%%%%%%%%%%%%%%%%%%%%%%%%%%%%%
A protocol $\Pi$ for problem $P$ is said to be \emph{correct} under the nFSM
model if
for every instance of $P$ and for every adversarial policy,
$\Pi$ reaches an output configuration within finite time with probability $1$,
and for every output configuration reached by $\Pi$ with positive probability,
the output of the nodes is a valid solution to $P$.
Given a correct protocol $\Pi$, the complexity measure that interests us in
the current paper is the \emph{run-time} of $\Pi$ defined as follows.

Consider some instance $\mathcal{I}$ of problem $P$.
Given an adversarial policy $\mathcal{A}$ and a sequence (actually an
$n$-tuple of sequences) $\mathcal{R}$ of random coin tosses that lead to an
output configuration within finite time, the run-time $T_{\Pi}(\mathcal{I},
\mathcal{A}, \mathcal{R})$ of $\Pi$ on $\mathcal{I}$ with respect to
$\mathcal{A}$ and $\mathcal{R}$ is defined as the (possibly fractional) number
of \emph{time units}\footnote{
Note that time units are defined solely for the purpose of the analysis.
Under an asynchronous environment, the nodes have no notion of time and in
particular, they cannot measure a single time unit.
} that pass from the beginning of the execution until the
first time the protocol reaches an output configuration, where a time unit is
defined to be the maximum among all step length parameters $L_{v, t}$ and
delivery delay parameters $D_{v, t, u}$ appearing in $\mathcal{A}$ before the
output configuration is reached.
Let $T_{\Pi}(\mathcal{I}, \mathcal{A})$ denote the random variable that
depicts the run-time of $\Pi$ on $\mathcal{I}$ with respect to $\mathcal{A}$.
Following the standard procedure in this regard, we say that the run-time of a
correct protocol $\Pi$ for problem $P$ is $f(n)$ if for every $n$-node
instance $\mathcal{I}$ of $P$ and for every adversarial policy $\mathcal{A}$,
it holds that $T_{\Pi}(\mathcal{I}, \mathcal{A})$ is at most $f(n)$ in
expectation and with high probability.
The protocol is said to be \emph{efficient} if its run-time is polylogarithmic
in the size of the network (cf. \cite{Linial92}).

%%%%%%%%%%%%%%%%%%%%%%%%%%%%%%%%%%%%%%%%%%%%%%%%%%%%%%%%%%%%%%%%%%%%%%%%%%%%%%
\section{Convenient Transformations}
\label{section:ConvenientTransformations}
%%%%%%%%%%%%%%%%%%%%%%%%%%%%%%%%%%%%%%%%%%%%%%%%%%%%%%%%%%%%%%%%%%%%%%%%%%%%%%

In this section, we show that the nFSM protocol designer may, in fact, assume
a slightly more ``user-friendly'' environment than the one described in
\Section{}~\ref{section:Model}.
This is based on the design of black-box compilers transforming a protocol
that makes strong assumptions on the environment into one that does not make
any such assumptions.
\LongVersion %{
Specifically, the assumptions that can be lifted that way are synchrony
(\Section{}~\ref{section:Synchronizer}), and multiple-letter queries
(\Section{}~\ref{section:MultiLetter}).
\LongVersionEnd %}

\LongVersion %{
%%%%%%%%%%%%%%%%%%%%%%%%%%%%%%%%%%%%%%%
\subsection{Implementing a Synchronizer}
\label{section:Synchronizer}
%%%%%%%%%%%%%%%%%%%%%%%%%%%%%%%%%%%%%%%
\LongVersionEnd %}
As described in \Section{}~\ref{section:Model}, the nFSM model assumes an
asynchronous environment.
Nevertheless, it will be convenient to extend the nFSM model to
\emph{synchronous} environments.
One natural such extension augments the model described in
\Section{}~\ref{section:Model} with the following two \emph{synchronization
properties} for every two adjacent nodes $u, v \in V$ and for every $t \in
\Integers_{> 0}$: \\
\textbf{(S1)}
when node $u$ is in step $t$, node $v$ is in step $t - 1$, $t$, or $t + 1$;
and \\
\textbf{(S2)}
at the end of step $t + 1$ of $u$, port $\psi_{u}(v)$ stores the message
transmitted by $v$ in step $t$ of $v$'s execution (or the last message
transmitted by $v$ prior to step $t$ if $v$ does not transmit any message in
step $t$). \\
An environment in which properties (S1) and (S2) are guaranteed to hold is
called a \emph{locally synchronous} environment.
Local-only communication can never achieve global synchrony, however, research
in the message passing model has shown that local synchrony is often
sufficient to provide efficient algorithms \cite{Awer85,AP90,APPS92}.
To distinguish a protocol assumed to operate in a locally synchronous
environment from those making no such assumptions, we shall often refer to the
execution steps of the former as \emph{rounds} (cf. fully synchronized
protocols).
\LongVersion %{
Our goal in this section is to establish the following theorem.
\LongVersionEnd %}

\begin{theorem} \label{theorem:Synchronizer}
Every nFSM protocol $\Pi = \langle Q, Q_I, Q_O, \Sigma, \sigma_0, b, \lambda,
\delta \rangle$ designed to operate in a locally synchronous environment can be
simulated in an asynchronous environment by a protocol $\widehat{\Pi}$ at the
cost of a constant multiplicative run-time overhead.
\end{theorem}

The procedure in charge of the simulation promised in
Theorem~\ref{theorem:Synchronizer}
\ShortVersion %{
(whose proof is deferred to \Appendix{}~\ref{appendix:Synchronizer})
\ShortVersionEnd %}
is referred to as a \emph{synchronizer} \cite{Awer85}.
\LongVersion %{
The remainder of \Section{}~\ref{section:Synchronizer} is dedicated to the
design (and analysis) of a synchronizer for the nFSM model.
\LongVersionEnd %}

\def\DetailsSynchronizer{
%%%%%%%%%%%%%%%%%%%%%%%%%%%%%%%%%%%%%%%
\paragraph{Overview.}
%%%%%%%%%%%%%%%%%%%%%%%%%%%%%%%%%%%%%%%
Round $t \in \Integers_{> 0}$ of node $v \in V$ under $\Pi$ is simulated by $O
(1)$ contiguous steps under $\widehat{\Pi}$;
the collection of these steps is referred to as $v$'s \emph{simulation phase}
of round $t$.
Protocol $\widehat{\Pi}$ is designed so that $v$ maintains the value of $t
\bmod 3$, referred to as the \emph{trit} (trinary digit) of round $t$, which
is also encoded in the message transmitted by $v$ at the end of round
$t$.\footnote{
Note that maintaining the value of $t \bmod 2$ is insufficient for the sake
of reaching synchronization.
}
The main principle behind our synchronizer is that node $v$ will not move to
the simulation phase of round $t + 1$ while its ports still contain messages
sent in a round whose trit is $t - 1 \bmod 3$.

Under $\Pi$, the decisions made by node $v$ at round $t$ should be based
on the messages transmitted by all neighbors $u$ of $v$ at round $t - 1$.
However, during $v$'s simulation phase of round $t$, port $\psi_{v}(u)$ may
contain messages transmitted at round $t - 1$ or at round $t$ under $\Pi$.
The latter case is problematic since the message transmitted by $u$ in
the simulation phase of round $t - 1$ is overwritten by that transmitted in the
simulation phase of round $t$.
To avoid this obstacle, a message transmitted by node $u$ under
$\widehat{\Pi}$ at the end of the simulation phase of round $t$ also
encodes the message that $u$ transmitted under $\Pi$ at round $t - 1$.

So, if $v$ resides in a state whose query letter is $\sigma \in \Sigma$ in
round $t$ under $\Pi$, then under $\widehat{\Pi}$, $v$ should query for all
$\widehat{\Sigma}$-letters encoding a transmission of $\sigma$ at round $t -
1$.
Since there are several such letters, a carefully designed feature should
be used so that $\widehat{\Pi}$ accounts for their combined number.

%%%%%%%%%%%%%%%%%%%%%%%%%%%%%%%%%%%%%%%
\paragraph{Protocol $\widehat{\Pi}$.}
%%%%%%%%%%%%%%%%%%%%%%%%%%%%%%%%%%%%%%%
Let
$$
\widehat{\Pi} ~ = ~
\left\langle \widehat{Q}, \widehat{Q}_{I}, \widehat{Q}_{O}, \widehat{\Sigma},
\widehat{\sigma}_{0}, b, \widehat{\lambda}, \widehat{\delta} \right\rangle \ .
$$
Consider node $v \in V$ and round $t \in \Integers_{> 0}$.
As the name implies, node $v$'s \emph{simulation phase} of round $t$ under
$\widehat{\Pi}$, denoted $\phi_{v}(t)$, corresponds to round $t$ of $\Pi$.
Protocol $\widehat{\Pi}$ is designed so that at every step in $\phi_{v}(t)$
other than the last one, $v$ does not transmit any message (indicated by
transmitting $\varepsilon$), and at the last step of the simulation phase, $v$
always transmits some message $\widehat{\sigma} \in \widehat{\Sigma}$, denoted
$M_{v}(t)$.

The alphabet $\widehat{\Sigma}$ is defined to be
$$
\Sigma' ~ = ~
\left( \Sigma \cup \{\varepsilon\} \right) \times
\left( \Sigma \cup \{\varepsilon\} \right) \times
\{ 0, 1, 2 \} \ .
$$
The semantics of the message $M_{v}(t) = (\sigma, \sigma', j)$ sent by
node $v$ at the last step of the simulation phase $\phi_{v}(t)$ is that:
$v$ transmits $\sigma \in \Sigma \cup \{\varepsilon\}$ at round $t - 1$ under
$\Pi$;
$v$ transmits $\sigma' \in \Sigma \cup \{\varepsilon\}$ at round $t$
under $\Pi$; and
$j = t \bmod 3$.
Following that logic, we set $\widehat{\sigma}_{0} = (\varepsilon, \sigma_{0},
0)$.

The state set $\widehat{Q}$ of $\widehat{\Pi}$ is defined to be
$$
\widehat{Q} ~ = ~
\left( \bigcup_{q \in Q} \left( P_{q} \cup S_{q} \right) \right) \times
\{ 0, 1, 2 \} \ ,
$$
where $P_{q} \times \{j\}$ and $S_{q} \times \{j\}$, $q \in Q$, $j \in \{ 0,
1, 2 \}$, are referred to as the \emph{pausing} and \emph{simulating} features,
respectively, whose role will be clarified soon.
Suppose that $v$ resides in state $q \in Q$ in step $t$ under $\Pi$ and that
$j = t \bmod 3$.
Then, throughout $\phi_{v}(t)$, node $v$ resides in some state in $(P_{q} \cup
S_{q}) \times \{j\}$.
In particular, in the first steps of the simulation phase, $v$ resides in
states of the pausing feature $P_{q} \times \{j\}$, and then at some stage it
switches to the simulating feature $S_{q} \times \{j\}$ and remains in its
states until the end of the simulation phase.

%%%%%%%%%%%%%%%%%%%%%%%%%%%%%%%%%%%%%%%
\paragraph{The Pausing Feature.}
%%%%%%%%%%%%%%%%%%%%%%%%%%%%%%%%%%%%%%%
For the simulation phase of round $t$, we denote the letters in $(\Sigma \cup
\{\varepsilon\}) \times (\Sigma \cup \{\varepsilon\}) \times \{ j - 2 \}$ as
\emph{dirty} and the letters in $(\Sigma \cup \{\varepsilon\}) \times (\Sigma
\cup \{\varepsilon\}) \times \{ j - 1, j \}$ as \emph{clean}.\footnote{
Throughout this section, arithmetic involving the parameter $j$ is done modulo
$3$.
}
The purpose of the pausing feature $P_{q} \times \{j\}$ is to pause the
execution of $v$ until its ports do not contain any dirty letter.
This is carried out by including in $P_{q} \times \{j\}$ a state $p_{\sigma,
\sigma'}$ for every $\sigma, \sigma' \in \Sigma \cup \{\varepsilon\}$;
the query letter of $p_{\sigma, \sigma'}$ is (the dirty letter)
$\widehat{\lambda}(p_{\sigma, \sigma'}) = (\sigma, \sigma', j - 2)$ and the
transition function $\widehat{\delta}$ is designed so that $v$ moves to the
next (according to some fixed order) state in the feature $P_{q} \times \{j\}$
if and only if there are no ports storing the query letter.

We argue that the pausing feature guarantees synchronization property (S1).
For the sake of the analysis, it is convenient to assume the existence of a
fully synchronous simulation phase of a virtual round $0$;
upon completion of this simulation phase (at the beginning of the execution),
every node $v \in V$ transmits the message $M_{v}(0) = \widehat{\sigma}_{0}$.
We are now ready to establish the following lemma.

\begin{lemma} \label{lemma:PausingFeature}
For every $t \in \Integers_{> 0}$, $v \in V$, and $u \in \Neighbors(v)$, when
$v$ completes the pausing feature of $\phi_{v}(t)$, port $\psi_{v}(u)$ stores
either $M_{u}(t - 1)$ or $M_{u}(t)$.
\end{lemma}
\begin{proof}
By induction on $t$.
The base case of round $t = 0$ holds by our assumption that $\phi_{v}(0)$ and
$\phi_{u}(0)$ are fully synchronous.
Assume by induction that the assertion holds for round $t - 1$.
Applying the inductive hypothesis to both $u$ and $v$, we conclude that
(1) when $v$ completes the pausing feature of $\phi_{v}(t - 1)$, port
$\psi_{v}(u)$ stores either $M_{u}(t - 2)$ or $M_{u}(t - 1)$; and
(2) when $u$ completes the pausing feature of $\phi_{u}(t - 1)$, port
$\psi_{u}(v)$ stores either $M_{v}(t - 2)$ or $M_{v}(t - 1)$.

Let $\tau_{u}$ and $\tau_{v}$ denote the times at which $u$ and $v$ complete
the pausing feature of $\phi_{u}(t)$ and $\phi_{v}(t)$, respectively.
Since $v$ cannot complete the pausing feature of $\phi_{v}(t)$ while
$M_{u}(t - 2)$ is still stored in $\psi_{v}(u)$, it follows that at time
$\tau_{v}$, port $\psi_{v}(u)$ stores the message $M_{u}(t')$ for some $t' \geq
t - 1$.
Our goal in the remainder of this proof is to show that $t' \leq t$.
If $\tau_{v} < \tau_{u}$, then $t'$ must be exactly $t - 1$, which concludes
the inductive step for that case.

So, assume that $\tau_{v} > \tau_{u}$ and suppose by contradiction that $t'
\geq t + 1$.
Using the same line of arguments as in the previous paragraph, we conclude
that at time $\tau_{u}$, port $\psi_{u}(v)$ stores the message $M_{v}(t - 1)$.
Node $u$ cannot complete the pausing feature of $\phi_{u}(t + 1)$ while
$M_{v}(t - 1)$ is still stored in $\psi_{u}(v)$, hence $v$ must have transmitted
$M_{v}(t)$ before $u$ completed the pausing feature of $\phi_{u}(t + 1)$.
But this means that $v$ completed the pausing feature of $\phi_{v}(t)$
before $u$ could have transmitted $M_{u}(t + 1)$, in contradiction to the
assumption that $\psi_{v}(u)$ stores $M_{u}(t')$ for some $t' \geq t + 1$ at
time $\tau_{v}$.
The assertion follows.
\end{proof}

Consider two adjacent nodes $u, v \in V$.
If node $u$ is at round $t - 1$ when an adjacent node $v$ is at round
$t + 1$, then $v$ completed the pausing feature of $\phi_{v}(t)$ before $u$
transmitted $M_{u}(t - 1)$, in contradiction to
Lemma~\ref{lemma:PausingFeature}.
Therefore, our synchronizer satisfies synchronization property (S1).
Furthermore, a similar argument shows that between the time $v$ completed the
pausing feature of $\phi_{v}(t)$ and the time $v$ completed the simulation
phase $\phi_{v}(t)$ itself, the content of $\psi_{v}(u)$ may change from
$M_{u}(t - 1)$ to $M_{u}(t)$ (if it was not already $M_{u}(t)$), but it will
not store $M_{u}(t')$ for any $t' > t$.
This fact is crucial for the implementation of the simulation feature.

%%%%%%%%%%%%%%%%%%%%%%%%%%%%%%%%%%%%%%%
\paragraph{The Simulation Feature.}
%%%%%%%%%%%%%%%%%%%%%%%%%%%%%%%%%%%%%%%
Upon completion of the pausing feature $P_{q} \times \{j\}$, $v$ moves on to
the simulation feature $S_{q} \times \{j\}$.
The purpose of this feature is to perform the actual simulation of round $t$
in $v$, namely, to determine the state (of $Q$) dominating the simulation phase
of the next round and the message transmitted when moving from the simulation
phase of the current round to that of the next round.

To see how this works out, suppose that $\lambda(q) = \sigma \in \Sigma$.
We would have wanted node $v$ to count (up to the bounding parameter $b$) the
number of occurrences of $\widehat{\Sigma}$-letters in its ports that
correspond to the transmission of $\sigma$ at round $t - 1$ under $\Pi$, that
is, the number of occurrences of letters in $\Gamma_{t - 1} \cup \Gamma_{t}$,
where
$$
\Gamma_{t - 1} ~ = ~
\left\{
\left( \sigma', \sigma, j - 1 \right)
\mid \sigma' \in \Sigma \cup \{\varepsilon\}
\right\}
\qquad \text{and} \qquad
\Gamma_{t} ~ = ~
\left\{
\left( \sigma, \sigma', j \right)
\mid \sigma' \in \Sigma \cup \{\varepsilon\}
\right\} \ .
$$
More formally, the application of the transition function $\widehat{\delta}$
at the end of the simulation phase $\phi_{v}(t)$ should be based on
$f_{b}(\sum_{\gamma \in \Gamma_{t - 1} \cup \Gamma_{t}} \Number(\gamma))$, where
$\Number(\gamma)$ stands for the number of occurrences of the letter $\gamma$
in the ports of $v$ at the end of $\phi_{v}(t)$.

Identifying the integer $b$ with the symbol $\BoundPar$, we observe that the
function $f_{b} : \Integers_{\geq 0} \rightarrow B$ satisfies
$$
f_{b}(x + y)
~ = ~
\min \left\{
f_{b}(x) + f_{b}(y), b
\right\}
$$
for every $x, y \in \Integers_{\geq 0}$.
A natural attempt to compute $f_{b}(\sum_{\gamma \in \Gamma_{t - 1} \cup \Gamma_{t}}
\Number(\gamma))$ would include in the feature $S_{q} \times \{j\}$ a state
$s_{\gamma, i}$ for every letter $\gamma \in \Gamma_{t - 1} \cup \Gamma_{t}$
and integer $i \in \{0, \dots, b\}$;
the query letter of $s_{\gamma, i}$ would be $\widehat{\lambda}(s_{\gamma}) =
\gamma$ and the transition function $\widehat{\delta}$ would be designed so
that $v$ moves from $s_{\gamma, i}$ to $s_{\gamma', i'}$, where $\gamma'$
follows $\gamma$ in some fixed order of the letters in $\Gamma_{t - 1} \cup
\Gamma_{t}$ and $i' = \min\{ i + f_{b}(\Number(\gamma)), b \}$.

However, care must be taken with this approach since $\Number(\gamma)$ may
decrease (respectively, increase) during $\phi_{v}(t)$ for $\gamma \in
\Gamma_{t - 1}$ (resp., for $\gamma \in \Gamma_{t}$) due to new incoming
messages.
To avoid this obstacle, we design the feature $S_{q} \times \{j\}$ so that 
first, it computes
$\varphi_1 \leftarrow f_{b}(\sum_{\gamma \in \Gamma_{t - 1}} \Number(\gamma))$;
next, it computes $\varphi_2 \leftarrow f_{b}(\sum_{\gamma \in \Gamma_{t}}
\Number(\gamma))$; and
finally, it computes ``again'' $\varphi_3 \leftarrow f_{b}(\sum_{\gamma \in
\Gamma_{t - 1}} \Number(\gamma))$.
If $\phi_1 = \phi_3$, then the current simulation phase is over and
$\widehat{\delta}$ is applied, simulating $\delta(q, f_{b}(\phi_1 + \phi_2))$;
otherwise, the feature $S_{q} \times \{j\}$ is invoked from scratch.
Since the value of $f_{b}(\sum_{\gamma \in \Gamma_{t - 1}} \Number(\gamma))$ cannot
increase during the simulation phase, and since $\phi_1 \leq b$, the feature
$S_{q} \times \{j\}$ is invoked at most $b$ times throughout the execution of
the simulation phase.
By induction on $t$, we conclude that our synchronizer satisfies
synchronization property (S2), which concludes the correctness proof of the
simulation.

%%%%%%%%%%%%%%%%%%%%%%%%%%%%%%%%%%%%%%%
\paragraph{Accounting.}
%%%%%%%%%%%%%%%%%%%%%%%%%%%%%%%%%%%%%%%
It remains to show that all ingredients of protocol $\widehat{\Pi}$ are of
constant size and that the run-time of protocol $\widehat{\Pi}$ incurs at most
a constant multiplicative overhead on top of that of protocol $\Pi$.
The former claim is established by following our synchronizer construction,
observing that $|\widehat{\Sigma}| = O \left( |\Sigma|^2 \right)$ and
$|\widehat{Q}| = O \left( |Q| \cdot (|\Sigma|^2 + |\Sigma| \cdot b) \right)$
(recall that the bounding parameter $b$ remains unchanged).
For the latter claim, we need the following definition:
given some node subset $U \in V$ and round $t \in \Integers_{> 0}$, let
$\tau(U, t)$ denote the first time at which $u$ completed simulation phase
$\phi_{u}(t)$ for all nodes $u \in U$.
The following proposition can now be established.

\begin{proposition} \label{proposition:SynchronizerOverhead}
For every node $v \in V$ and round $t \in \Integers_{> 0}$, the time difference
$\tau(\{v\}, t + 1) - \tau(\Neighbors(v) \cup \{v\}, t)$ is (up)bounded by a
constant.
\end{proposition}
\begin{proof}
Since each transmitted message has a delay of at most $1$ unit of time, it
follows that by time $\tau(\Neighbors(v) \cup \{v\}, t) + 1$, message
$M_{u}(t)$ must reach $\psi_{v}(u)$ for all $u \in \Neighbors(v)$.
The pausing and simulation features of $\phi_{v}(t + 1)$ are then completed
within $O (|\Sigma|^2)$ and $O (|\Sigma| \cdot b)$ steps, respectively.
The assertion follows as each step lasts for at most $1$ unit of time.
\end{proof}

Employing Proposition~\ref{proposition:SynchronizerOverhead}, we conclude by
induction on $t$ that $\tau(V, t) = O (t)$ for every $t \in \Integers_{> 0}$,
hence if the execution of protocol $\Pi$ requires $T$ rounds, then the
execution of protocol $\widehat{\Pi}$ is completed within $O (T)$ time units.
Theorem~\ref{theorem:Synchronizer} follows.
} % end \DetailsSynchronizer
\LongVersion %{
\DetailsSynchronizer{}
\LongVersionEnd %}

\LongVersion %{
%%%%%%%%%%%%%%%%%%%%%%%%%%%%%%%%%%%%%%%
\subsection{Multiple-Letter Queries}
\label{section:MultiLetter}
%%%%%%%%%%%%%%%%%%%%%%%%%%%%%%%%%%%%%%%
\LongVersionEnd %}
Recall that according to the model presented in \Section{}~\ref{section:Model},
each state $q \in Q$ is associated with a query letter $\lambda(q)$ and the
application of the transition function when node $v$ resides in state $q$ is
determined by $f_{b}(\Number(\sigma))$, where $\Number(\sigma)$ is the number of
occurrences of the letter $\sigma$ in the ports of $v$.
From the perspective of the protocol designer, it is often more convenient to
assume that the node queries on all letters simultaneously, namely, that
the application of the transition function is determined by the vector
$\left\langle f_{b}(\Number(\sigma)) \right\rangle_{\sigma \in \Sigma}$.

Now that we may assume a synchronous environment, this stronger
multiple-letter queries assumption can easily be supported.
Indeed, at the cost of increasing the number of states and the run-time by
constant factors, one can subdivide each round into $|\Sigma|$ subrounds,
dedicating each subround to a different letter in $\Sigma$, so that at the end
of the round, the state of $v$ reflects $f_{b}(\Number(\sigma))$ for every
$\sigma \in \Sigma$.

\begin{theorem} \label{theorem:MultiLetter}
Every nFSM protocol with multiple-letter queries can be simulated by an nFSM
protocol with single-letter queries at the cost of a constant multiplicative
run-time overhead.
\end{theorem}

%%%%%%%%%%%%%%%%%%%%%%%%%%%%%%%%%%%%%%%%%%%%%%%%%%%%%%%%%%%%%%%%%%%%%%%%%%%%%%
\section{Maximal Independent Set}
\label{section:MIS}
%%%%%%%%%%%%%%%%%%%%%%%%%%%%%%%%%%%%%%%%%%%%%%%%%%%%%%%%%%%%%%%%%%%%%%%%%%%%%%
Given a graph $G = (V, E)$, the \emph{maximal independent set (MIS)} problem
asks for a node subset $U \subseteq V$ which is independent in the sense
that $(U \times U) \cap E = \emptyset$, and maximal in the sense that $U'
\subseteq V$ is not independent for every $U' \supset U$.
Distributed MIS algorithms with logarithmic run-time operating in the
message passing model were presented by Luby~\cite{Luby86} and independently,
by Alon et al.~\cite{ABI86};\footnote{
The focus of \cite{Luby86} and \cite{ABI86} was actually on the PRAM model,
but their algorithms can be adapted to the message passing model.
}
Luby's algorithm has since become a specimen of distributed algorithms;
in the last 25 years, researchers have tried to improve it, if only e.g., with
an improved bit complexity~\cite{MRSZ11}, on special graph
classes~\cite{SW10,LW11}, or in a weaker communication
model~\cite{AfekAlonBarJoseph+11}.
An $\Omega (\sqrt{\log n})$-lower bound on the run-time of any distributed
MIS algorithm operating in the message passing model was established by Kuhn
et al.~\cite{KMW04}.
Our goal in this section is to design an nFSM protocol for the MIS problem
with run-time $O (\log^2 n)$.

%%%%%%%%%%%%%%%%%%%%%%%%%%%%%%%%%%%%%%%
\paragraph{Outline of the Key Technical Ideas.}
%%%%%%%%%%%%%%%%%%%%%%%%%%%%%%%%%%%%%%%
Our protocol is inspired by the existing message passing MIS algorithms.
Common to all these algorithms is that they are based on the concept of
grouping consecutive rounds into \emph{phases}, where in each phase, nodes
compete against their neighbors over the right to join the MIS.
Existing implementations of such competitions require at least one of the
following three capabilities:
(1) performing calculations that involve super-constant numbers;
(2) communicating with each neighbor independently; or
(3) sending messages of super-constant size, specifically, of size $c \log n$
for some constant $c > 0$.
The first two capabilities are clearly out of the question for an nFSM
protocol.
The third one is also not supported by the nFSM model, but perhaps one can
divide a message with a logarithmic number of bits over logarithmic many
rounds, sending $1$ (or $O (1))$ bits per round (cf. Algorithm~B in
\cite{MRSZ11})?

This naive attempt results in phases of length $c \log n$.
However, no FSM can count the rounds in a $c \log n$ long phase --- a task
essential for deciding if the current phase is over and the next one
should begin.
Furthermore, to guarantee fair competition, the phases must be aligned
across the network, thus ruling out the possibility to start node $v$'s phase
$i$ before phase $i - 1$ of some node $u \in \Neighbors(v)$ is finished.
In fact, an efficient algorithm that requires $\omega (1)$ long aligned
phases cannot be implemented under the nFSM model.
So, how can we decide if node $v$ joins the MIS using constant size messages
without the ability to maintain long aligned phases?

This issue is resolved by relaxing the requirements that the phases are
aligned and of a predetermined length, introducing a feature referred to
as a \emph{tournament}.
Our tournaments are only ``softly'' aligned and their lengths are
determined probabilistically, in a manner that can be maintained under the
nFSM model.
Nevertheless, they enable a fair competition between neighboring nodes, as
desired.

\LongVersion %{
\begin{AvoidOverfullParagraph}
\LongVersionEnd %}
%%%%%%%%%%%%%%%%%%%%%%%%%%%%%%%%%%%%%%%
\paragraph{The Protocol.}
%%%%%%%%%%%%%%%%%%%%%%%%%%%%%%%%%%%%%%%
Employing Theorems \ref{theorem:Synchronizer} and \ref{theorem:MultiLetter},
we assume a locally synchronous environment and use multiple-letter queries.
The state set of the protocol is
$Q = \{ \Win, \Lose, \Down_{1}, \Down_{2}, \Up_{0}, \Up_{1}, \Up_{2} \}$,
with $Q_{I} = \{ \Down_{1} \}$ (the initial state of all nodes) and $Q_{O} = \{
\Win, \Lose \}$, where $\Win$ (respectively, $\Lose$) indicates membership
(resp., non-membership) in the MIS output by the protocol.
The states in $Q_{A} = Q - Q_{O}$ are called the \emph{active} states and a
node in an active state is referred to as an \emph{active} node.
We take the communication alphabet $\Sigma$ to be identical to the state set
$Q$, where the letter transmissions are designed so that node $v$ transmits
letter $q$ whenever it moves to state $q$ from some state $q' \neq q$;
no letter is transmitted in a round at which $v$ remains in the same state.
Letter $\Down_{1}$ is the initial letter stored in all ports at the beginning
of the execution.
The bounding parameter is set to $b = 1$.
\LongVersion %{
\end{AvoidOverfullParagraph}
\LongVersionEnd %}

A schematic description of the transition function is provided in
\Figure{}~\ref{figure:MisFsm};
its logic is as follows.
Each state $q \in Q_{A}$ has a subset $D(q) \subseteq Q_{A}$ of \emph{delaying
states}:
node $v$ remains in the current state $q$ as long as (at least) one of its
neighbors is in some state in $D(q)$.
This is implemented by querying on the letters (corresponding to the states)
in $D(q)$, staying in state $q$ as long as at least one of these letters is
found in the ports.
Specifically, state $\Down_{1}$ is delayed by state $\Down_{2}$, which is
delayed by all three $\Up$ states.
State $\Up_{j}$, $j = 0, 1, 2$, is delayed by state $\Up_{j - 1 \bmod 3}$,
where state $\Up_{0}$ is also delayed by state $\Down_{1}$.

\tikzstyle{state with output}=[ellipse split,draw=black, minimum width=12mm,
minimum height=11mm, inner sep=0.4mm]

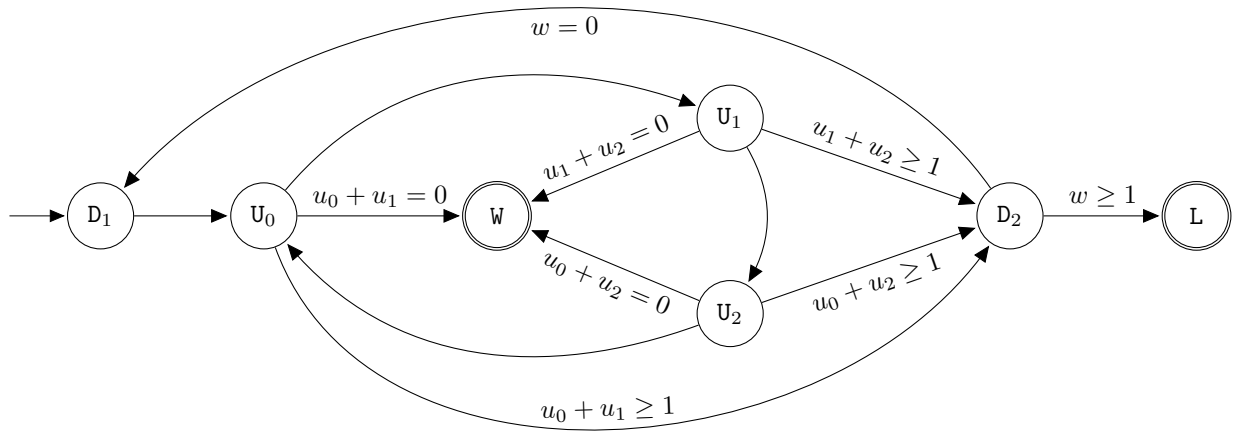
\begin{figure}[!t]
\begin{center}
\begin{small}
\begin{tikzpicture}[shorten >=1pt, auto, initial text=, x=3.1cm, y=1.3cm,
>=triangle 45, initial/.style={pin={[pin distance=8mm, pin edge={<-,shorten
<=1pt, black}]left:}}]

\node at (0.3,1.5) [state, initial] (down1) {$\DownShort_{1}$};

\node at (1,1.5) [state] (up0) {$\UpShort_{0}$};

\node at (3,2.5) [state] (up1) {$\UpShort_{1}$};

\node at (3,0.5) [state] (up2) {$\UpShort_{2}$};

\node at (4.2,1.5) [state] (down2) {$\DownShort_{2}$};

\node at (5,1.5) [state, accepting] (lose) {$\LoseShort$};

\node at (2,1.5) [state, accepting] (win) {$\WinShort$};

\path[->]
(up0) edge node [sloped, midway, above] {$u_0+u_1=0$} (win)

(up0) edge [out=290, in=235, looseness=1.0] node [sloped, midway, above]
{$u_0+u_1\geq1$} (down2)

(up0) edge [out=50, in=160] node [sloped, midway, above] {} (up1)

(up1) edge node [sloped, midway, above]{$u_1+u_2=0$} (win)

(up1) edge node [sloped, midway, above] {$u_1+u_2\geq 1$} (down2)

(up1) edge [bend left] node [sloped, midway, above] {} (up2)

(up2) edge node [sloped, midway, below] {$u_0+u_2=0$} (win)

(up2) edge node [sloped, midway, below] {$u_0+u_2\geq 1$} (down2)

(up2) edge [out=200, in=-50] node [sloped, midway, below] {} (up0)

(down1) edge node [sloped, midway, above] {} (up0)

(down2) edge [out=125, in=50, looseness=0.9] node [sloped, midway, below]
{$w=0$} (down1)

(down2) edge node [sloped, midway, above] {${w\geq 1}$} (lose)
;
\end{tikzpicture}
\end{small}
\end{center}
\caption{\label{figure:MisFsm}
The transition function of the MIS protocol with state names abbreviated by
their first (capital) letters.
The node stays in state $q$ (a.k.a. \emph{delayed}) as long as $\Number(q') >
0$ for any state $q'$ such that a $q' \rightarrow q$ transition is defined
(for clarity, this is omitted from the figure).
Assuming that the node is not delayed, each transition specified in the figure
is associated with a condition on the number of appearances of the query
letters in the ports (depicted by the corresponding lower-case letter) so that
the transition is followed only if the condition is satisfied (an empty
condition is satisfied by all port configurations);
if some port configuration satisfies several transition conditions, then one
of them is chosen uniformly at random.
}
\end{figure}

States $\Win$ and $\Lose$ are sinks in the sense that a node that moves
to one of these states will stay there indefinitely.
Assuming that node $v$ does not find any delaying letter in its ports, the
logic of the $\Up$ and $\Down$ states is as follows.
From state $\Down_{1}$, $v$ moves to state $\Up_{0}$.
From state $\Down_{2}$, $v$ moves to state $\Down_{1}$ if $\Number(\Win) =
0$, that is, if it does not find any $\Win$ letter in its ports;
otherwise, it moves to state $\Lose$.
When in state $\Up_{j}$, $v$ tosses a fair coin and proceeds as follows:
if the coin turns head, then $v$ moves to state $\Up_{j + 1 \bmod 3}$;
if the coin turns tail, then $v$ moves to state $\Win$ if $\Number(\Up_{j}) =
\Number(\Up_{j + 1 \bmod 3}) = 0$;
and to state $\Down_{2}$ otherwise.
This completes the description of our nFSM protocol for the MIS problem.

%%%%%%%%%%%%%%%%%%%%%%%%%%%%%%%%%%%%%%%
\paragraph{Turns and Tournaments.}
%%%%%%%%%%%%%%%%%%%%%%%%%%%%%%%%%%%%%%%
Our protocol is designed so that an active node $v$ traverses the $\Down$ and
$\Up$ states in a (double-)circular fashion:
an inner loop of the $\Up$ states (moving from state $\Up_{j}$ to state
$\Up_{j + 1 \bmod 3}$) nested within an outer loop consisting of the $\Down$
states and the inner loop.
Of course, $v$ may spend more than one round at each state $q \in Q_{A}$
(delayed by adjacent nodes in states $D(q)$);
we refer to a maximal contiguous sequence of rounds that $v$ spends in the same
state $q \in Q_{A}$ as a \emph{$q$-turn}, or simply as a
\emph{turn} if the actual state $q$ is irrelevant.
A maximal contiguous sequence of turns that starts at a $\Down_{1}$-turn and
does not include any other $\Down_{1}$-turn (i.e., a single iteration of the
outer loop) is referred to as a \emph{tournament}.
We index the tournaments and the turns within a tournament by the positive
integers.
Note that by definition, every tournament $i$ of $v$ starts with a
$\Down_{1}$-turn, followed by a non-empty sequence of $\Up$-turns.
If tournament $i + 1$ of $v$ exists, then tournament $i$ ends with a
$\Down_{2}$-turn;
otherwise, it ends with an $\Up$-turn.
The following observation is established by induction on the rounds.
\begin{observation} \label{observation:NeighborsTurns}
Consider some node $v \in V$ in turn $j \in \Integers_{> 0}$ of tournament $i
\in \Integers_{> 0}$ and some active node $u \in \Neighbors(v)$.
\begin{DenseItemize}

\item
If this is a $\Down_{1}$-turn of $v$ ($j = 1$), then $u$ is in either
(A) the last ($\Down_{2}$-)turn of tournament $i - 1$;
(B) turn $1$ of tournament $i$; or
(C) turn $2$ of tournament $i$.

\item
If this is an $\Up$-turn of $v$ ($j \geq 2$), then $u$ is in either
(A) turn $j - 1$ of tournament $i$;
(B) turn $j$ of tournament $i$;
(C) turn $j + 1$ of tournament $i$; or
(D) the last ($\Down_{2}$-)turn $j' \leq j + 1$ of tournament $i$.

\item
If this is a $\Down_{2}$-turn of $v$ (the last turn of this tournament),
then $u$ is in either
(A) an $\Up$-turn $j' \geq j - 1$ of tournament $i$;
(B) the last ($\Down_{2}$-)turn of tournament $i$; or
(C) turn $1$ of tournament $i + 1$.

\end{DenseItemize}
\end{observation}

Given some $U \subseteq V$ and $i, j \in \Integers_{> 0}$, let $T_{U}(i, j)$
denote the first time at which every node $v \in U$ satisfies either
\LongVersion \\ \LongVersionEnd
(1) $v$ is inactive; \LongVersion \\ \LongVersionEnd
(2) $v$ is in tournament $i' > i$; \LongVersion \\ \LongVersionEnd
(3) $v$ is in the last ($\Down_{2}$-)turn of tournament $i$; or \LongVersion \\
\LongVersionEnd
(4) $v$ is in turn $j' \geq j$ of tournament $i$. \LongVersion \\
\LongVersionEnd
Employing Observation~\ref{observation:NeighborsTurns}, the delaying states
feature guarantees that
\begin{equation} \label{equation:TurnsProgressSingleNode}
T_{v}(i, j + 1)
~ \leq ~
T_{\Neighbors(v) \cup \{v\}}(i, j) + 1
\end{equation}
for every $v \in V$ and $i, j \in \Integers_{> 0}$.
Since $T_{U}(i, j) \leq T_{V}(i, j)$ for every $U \subseteq V$, we can apply
inequality~(\ref{equation:TurnsProgressSingleNode}) to each node $v \in V$,
concluding that
\[
T_{V}(i, j + 1)
~ \leq ~
T_{V}(i, j) + 1 \, ,
\]
which immediately implies that
\begin{equation} \label{equation:TurnsProgressAllNodes}
T_{V}(i, k + 1)
~ \leq ~
T_{V}(i, 1) + k \, .
\end{equation}

%%%%%%%%%%%%%%%%%%%%%%%%%%%%%%%%%%%%%%%
\paragraph{Geometric Random Variables.}
%%%%%%%%%%%%%%%%%%%%%%%%%%%%%%%%%%%%%%%
Consider some $v \in V$ and $i \in \Integers_{> 0}$.
Assuming that tournament $i$ of $v$ exists, let $X_{v}(i)$ denote its
\emph{length} in terms of number of turns.
For the sake of simplifying the analysis, if tournament $i$ is the last
tournament of $v$, then we actually take $X_{v}(i)$ to be its length plus $1$
(this is done in order to compensate for the missing $\Down_{2}$-turn in the
end of the tournament.)
The logic of the $\Up$ states implies that $X_{v}(i)$ is a random variable
that obeys distribution $\Geom(1 / 2) + 2$, namely, a fixed term of $2$ plus
the geometric distribution with parameter $1 / 2$, independently of
$X_{v'}(i')$ for any $v' \neq v$ and/or $i' \neq i$.
Since the maximum of $n$ independent $\Geom(1 / 2)$-random variables is $O
(\log n)$ with high probability,
inequality~\ref{equation:TurnsProgressAllNodes} yields the following
observation.

\begin{observation} \label{observation:TournamentsProgress}
For every $i \in \Integers_{> 0}$, $T_{V}(i, 1)$ is finite with probability
$1$ and
\[
T_{V}(i + 1, 1)
~ \leq ~
T_{V}(i, 1) + O (\log n)
\]
with high probability.
\end{observation}

Our protocol is designed so that node $v$ moves to an output state ($\Win$ or
$\Lose$) in the end of each tournament with positive probability.
Moreover, the logic of state $\Down_{2}$ guarantees that if node $v$ moves to
state $\Win$ in the end of tournament $i$, then all its active neighbors move
to state $\Lose$ in the end of their respective tournaments $i$.
By Observation~\ref{observation:TournamentsProgress}, we conclude that our
protocol reaches an output configuration with probability $1$ and that every
output configuration reflects an MIS.
It remains to bound the run-time of our protocol.

%%%%%%%%%%%%%%%%%%%%%%%%%%%%%%%%%%%%%%%
\paragraph{The Virtual Graph $G^{i}$.}
%%%%%%%%%%%%%%%%%%%%%%%%%%%%%%%%%%%%%%%
Let $V^{i}$ be the set of nodes for which tournament $i$ exists and let $G^{i}
= (V^{i}, E^{i})$ be the subgraph induced on $G$ by $V^{i}$, where $E^{i}
= E \cap (V^{i} \times V^{i})$.\footnote{
The notation $G^{i}$ used in this section should not be confused with the
$i^{\text{th}}$ power of $G$.
}
Given some node $v \in V^{i}$, let $\Neighbors^{i}(v) = \{ u \in V^{i} \mid
(u, v) \in E \}$ be the neighborhood of node $v$ in $G^{i}$ and let $d^{i}(v)
= |\Neighbors^{i}(v)|$ be its degree.
Note that the graph $G^{i}$ is virtual and defined solely for the sake of the
analysis;
in particular, we do not assume that there exists some time at which the graph
induced by any meaningful subset of the nodes (say, the nodes in tournament
$i$) agrees with $G^{i}$.
The key observation in this context is that conditioned on $G^{i}$, the random
variables $X_{v}(i)$, $v \in V^{i}$, are (still) independent and obey
distribution $\Geom(1 / 2) + 2$.
Moreover, the graph $G^{i + 1}$ is fully determined by the random
variables $X_{v}(i)$, $v \in V^{i}$.
Our analysis relies on the following lemma.

\begin{lemma} \label{lemma:ExponentialSizeDecrease}
There exist two constants $0 < p, c < 1$ such that $|E^{i + 1}| \leq c
|E^{i}|$ with probability at least $p$.
\end{lemma}

We will soon turn to proving Lemma~\ref{lemma:ExponentialSizeDecrease}, but
first, let us explain why it suffices for the completion of our analysis.
Define the random variable $Y = \min\{ i \in \Integers_{> 0} : |E^{i}| = 0
\}$.
Lemma~\ref{lemma:ExponentialSizeDecrease} implies that $Y$ is stochastically
dominated by a random variable that obeys distribution $\NegativeBinomial(O
(\log n), 1 - p) + O (\log n)$, namely, a fixed term of $O (\log n)$ plus the
negative binomial distribution with parameters $O (\log n)$ and $1 - p$, hence
$Y = O (\log n)$ in expectation and with high probability.
Since the nodes in $V - V^{i}$ are all in an output state (and will remain in
that state), and since the logic of the $\Up$ states implies that a degree-$0$
node in $G^{i}$ will move to state $\Win$ in the end of tournament $i$ (with
probability $1$) and thus, will not be included in $V^{i + 1}$, we can employ
Observation~\ref{observation:TournamentsProgress} to conclude that the
run-time of our protocol is $O (\log^2 n)$.

The remainder of this section is dedicated to establishing
Lemma~\ref{lemma:ExponentialSizeDecrease}.
The proof technique we use for that purpose resembles (a hybrid of) the
techniques used in \cite{ABI86} and \cite{MRSZ11} for the analysis of their
MIS algorithms.
We say that node $v \in V^{i}$ is \emph{good} in $G^{i}$ if
\[
|\{ u \in \Neighbors^{i}(v) \mid d^{i}(u) \leq d^{i}(v) \}| \geq
d^{i}(v) / 3  \, ,
\]
i.e., if at least third of $v$'s neighbors in $G^{i}$ have degrees smaller or
equal to that of $v$.
The following lemma is established in \cite{ABI86}.

\begin{lemma}[\cite{ABI86}] \label{lemma:GoodEdges}
More than half of the edges in $E^{i}$ are incident on good nodes in $G^{i}$.
\end{lemma}

%%%%%%%%%%%%%%%%%%%%%%%%%%%%%%%%%%%%%%%
\paragraph{Disjoint Winning Events.}
%%%%%%%%%%%%%%%%%%%%%%%%%%%%%%%%%%%%%%%
Consider some good node $v$ in $G^{i}$ with $d = d^{i}(v) > 0$ and let
$\widehat{\Neighbors}^{i}(v) = \{ u \in \Neighbors^{i}(v) \mid d^{i}(u) \leq
d \}$.
Recall that the definition of a good node implies that
$|\widehat{\Neighbors}^{i}(v)| \geq d / 3$.
We say that node $u \in \widehat{\Neighbors}^{i}(v)$ \emph{wins} $v$ in
tournament $i$ if
\[
X_{u}(i)
>
\max \left\{ X_{w}(i) \mid w \in \Neighbors^{i}(u) \cup
\widehat{\Neighbors}^{i}(v) - \{u\} \right\}
\]
and denote this event by $A^{i}(u, v)$.
The main observation now is that if $u$ wins $v$ in tournament $i$, then in
the end of their respective tournaments $i$, $u$ moves to state $\Win$ and $v$
moves to state $\Lose$.
Moreover, the events $A^{i}(u, v)$ and $A^{i}(w, v)$ are disjoint for every
$u, w \in \widehat{\Neighbors}^{i}(v)$, $u \neq w$.

Let $u_{1}, \dots, u_{k}$ be the nodes in
$\Neighbors^{i}(u) \cup \widehat{\Neighbors}^{i}(v)$,
where $0 < k \leq 2 \, d$ by the definition of a good node.
Let $B^{i}(u, v)$ denote the event that the maximum of
$\{ X_{u_{\ell}}(i) \mid 1 \leq \ell \leq k \}$
is attained at a single $1 \leq \ell \leq k$.
Since $X_{u_{1}}(i), \dots, X_{u_{k}}(i)$ are independent random variables that
obey distribution $\Geom(1 / 2) + 2$, it follows that
$\Probability(B^{i}(u, v)) \geq 2 / 3$.
Therefore,
\[
\Probability \left( A^{i}(u, v) \right) =
\Probability \left( A^{i}(u, v) \mid B^{i}(u, v) \right) \cdot
\Probability \left( B^{i}(u, v) \right) \geq
\frac{1}{k} \cdot \frac{2}{3} \, ,
\]
which implies that
\begin{align*}
\Probability \left( v \notin V^{i + 1} \mid v \text{ is good in } G^{i}
\right)
~ \geq ~ &
\Probability \left( \bigvee_{u \in \widehat{\Neighbors}^{i}(v)} A^{i}(u, v)
\right) \\
= ~ &
\sum_{u \in \widehat{\Neighbors}^{i}(v)} \Probability \left( A^{i}(u, v)
\right)
~ \geq ~
\frac{d}{3} \cdot \frac{1}{2 \, d} \cdot \frac{2}{3}
~ = ~
\frac{1}{9} \, .
\end{align*}
Combined with Lemma~\ref{lemma:GoodEdges}, we conclude that
$\Expectation[|E^{i + 1}|] < \frac{35}{36} \, |E^{i}|$.
Lemma~\ref{lemma:ExponentialSizeDecrease} follows by Markov's bound.

\begin{theorem}
There exists an nFSM protocol that computes an MIS in any $n$-node graph with
run-time $O (\log^2 n)$.
\end{theorem}

\LongVersion %{
%%%%%%%%%%%%%%%%%%%%%%%%%%%%%%%%%%%%%%%%%%%%%%%%%%%%%%%%%%%%%%%%%%%%%%%%%%%%%%
\section{Coloring a Tree with $3$ Colors}
\label{section:Coloring}
%%%%%%%%%%%%%%%%%%%%%%%%%%%%%%%%%%%%%%%%%%%%%%%%%%%%%%%%%%%%%%%%%%%%%%%%%%%%%%
\LongVersionEnd %}
\def\SectionTreeColoring{
Given a graph $G = (V, E)$, the \emph{coloring} problem asks for an assignment
of colors to the nodes such that no two neighboring nodes have the same color.
A coloring using at most $k$ colors is called a \emph{$k$-coloring}.
The smallest number of colors needed to color graph $G$ is called its
\emph{chromatic number}, denoted by $\chi(G)$.
In general, $\chi(G)$ is difficult to compute even in a centralized model
\cite{BellareGS98}.
As such, the distributed computing community is generally satisfied already
with a $(\Delta + 1)$-, $O(\Delta)$-, or even $\Delta^{O (1)}$-coloring, where
$\Delta = \Delta(G)$ is the largest degree in the graph $G$, with possibly
$\Delta(G) \gg \chi(G)$
\cite{CV86,Plot88,GPS88,Linial92,SV93,BE09,Kuhn09,BE11a,BE11b,Schneider2011ChromaticNumber}.
However, even for relatively simple graph classes, $\Delta$ may grow with $n$.
As the output of each node under the nFSM model is taken from a constant
size set, we must and will tackle a graph class that features a small
chromatic number: trees.

Any tree $T$ has a chromatic number $\chi(T) = 2$.
Unfortunately, it is easy to show that in general, the task of $2$-coloring
trees requires run-time proportional to the diameter of the tree even under
the message passing model, and hence cannot be achieved by an efficient
distributed algorithm.
The situation improves dramatically once $3$ colors are allowed;
indeed, Cole and Vishkin~\cite{CV86} presented a distributed algorithm that
$3$-colors directed paths, and in fact, any directed tree (directed in the
sense that each node knows the port leading to its unique parent), in time $O
(\log^{*} n)$.
Linial~\cite{Linial92} showed that this is asymptotically optimal. 

Since it is not clear how to represent directed trees in the nFSM model, we
focus on undirected trees, designing an nFSM protocol that $3$-colors any
$n$-node (undirected) tree in run-time $O (\log n)$.
A lower bound result of Kothapalli et al.~\cite{KothapalliSOS06} shows that
this cannot be improved (asymptotically) even by a message passing algorithm
as long as the size of each message is $O (1)$.

Employing Theorems \ref{theorem:Synchronizer} and \ref{theorem:MultiLetter},
we assume a locally synchronous environment and use multiple-letter queries.
The description of the protocol will not dwell into the level of defining the
states and transition function (as we did in \Section{}~\ref{section:MIS} for
the MIS protocol), but the reader will be easily convinced that this protocol
can indeed be implemented under the nFSM model.

%%%%%%%%%%%%%%%%%%%%%%%%%%%%%%%%%%%%%%%
\paragraph{The Modes.}
%%%%%%%%%%%%%%%%%%%%%%%%%%%%%%%%%%%%%%%
At all times, each node $v \in V$ is in one of the following three
\emph{modes}.
\\
(1) Mode $\Colored$: the color of $v$ is determined ($v$ is in an output
state) and it no longer takes an active part in the protocol. \\
(2) Mode $\Active$: the color of $v$ has not been determined yet and $v$ takes
an active part in the protocol. \\
(3) Mode $\Waiting$: the color of $v$ has not been determined yet and $v$ is
waiting for one of its neighbors to be colored before it resumes taking an
active part in the protocol (going back to mode $\Active$). \\

Initially, all nodes are in mode $\Active$.
When an $\Active$ node moves to mode $\Colored$, assigned with color $c \in \{
1, 2, 3 \}$, it transmits a `my color is $c$' message and it does not transmit
any more messages;
when an $\Active$ node moves to mode $\Waiting$, it transmits an `I am
$\Waiting$' message and it does not transmit any more messages until it
returns to mode $\Active$, in which case it transmits an `I am $\Active$'
message.
Therefore, the message stored in the port of node $v$ corresponding to neighbor
$u$ of $v$ always indicates (perhaps among other things) the current mode of
$u$.

%%%%%%%%%%%%%%%%%%%%%%%%%%%%%%%%%%%%%%%
\paragraph{The Phases.}
%%%%%%%%%%%%%%%%%%%%%%%%%%%%%%%%%%%%%%%
The execution of the protocol is divided into \emph{phases} indexed by the
positive integers, where each phase consists of $4$ rounds.
Consider some phase $i \in \Integers_{> 0}$.
Let $V^{i}$ be the set of $\Active$ nodes at the beginning of phase $i$ and
let $F^{i}$ be the forest induced on $T$ by $V^{i}$ ($F^{i}$ may contain one
or more trees), referred to as the \emph{$\Active$ forest}.
Given some node $v \in V^{i}$, let $\Neighbors^{i}(v) = \{ u \in V^{i} \mid
(u, v) \in E \}$ be the neighborhood of $v$ in $F^{i}$ and let $d^{i}(v)
= |\Neighbors^{i}(v)|$ be its degree.

The structure of the phases is as follows.
Consider some node $v \in V^{i}$.
In round $1$ of the phase, $v$ transmits an `I am $\Active$' message.
Setting the bounding parameter of the protocol to $b = 3$, we conclude that
in round $2$, $v$ can distinguish between the cases
$d^{i}(v) = 0$, $d^{i}(v) = 1$, $d^{i}(v) = 2$, and $d^{i}(v) \geq 3$
simply by querying its ports for `I am $\Active$' messages;
in other words, $v$ ``knows'' $f_{3}(d^{i}(v))$, i.e., its degree calculated
with respect to the one-two-many principle with bounding parameter $b = 3$.
Employing this ``knowledge'', $v$ transmits $f_{3}(d^{i}(v))$ in round $2$ of
phase $i$, so in round $3$, the port of $v$ corresponding to $u$ stores a
message indicating $f_{3}(d^{i}(u))$ for every node $u \in \Neighbors^{i}(v)$.

Rounds $3$ and $4$ of phase $i$ are dedicated to Procedure~$\RandomColoring$
that we will describe soon.
Whether or not $v$ runs Procedure~$\RandomColoring$ depends on the degree of
$v$ and on the degrees of its $\Active$ neighbors.
Specifically, $v$ runs Procedure~$\RandomColoring$ if:
(1) $d^{i}(v) = 0$;
(2) $d^{i}(v) = 1$ with $\Neighbors^{i}(v) = \{ u \}$ and $d^{i}(u) = 1$; or
(3) $d^{i}(v) = 2$ with $\Neighbors^{i}(v) = \{ u_1, u_2 \}$ and $d^{i}(u_1),
d^{i}(u_2) \leq 2$.
In contrast, if $d^{i}(v) = 1$ with $\Neighbors^{i}(v) = \{ u \}$ and
$d^{i}(u) \geq 2$, then $v$ moves to mode $\Waiting$ without running
Procedure~$\RandomColoring$, in which case we say (just for the sake of the
analysis) that $v$ \emph{waits on} $u$.
Otherwise ($d^{i}(v) \geq 3$ or $d^{i}(v) = 2$ with some neighbor $u \in
\Neighbors^{i}(v)$ such that $d^{i}(u) \geq 3$), $v$ remains in mode $\Active$
without running Procedure~$\RandomColoring$.

As stated beforehand, the $\Colored$ nodes do not take an active part in the
protocol.
A $\Waiting$ node $v$ moves to mode $\Active$ in the end of phase $i$ if some
neighbor $u$ of $v$, $u \in V^{i}$, moves to mode $\Colored$ during phase $i$
($v$ spots this event by querying on `my color is $c$' messages).

%%%%%%%%%%%%%%%%%%%%%%%%%%%%%%%%%%%%%%%
\paragraph{Procedure~$\RandomColoring$.}
%%%%%%%%%%%%%%%%%%%%%%%%%%%%%%%%%%%%%%%
Responsible for the actual color assignments, Procedure~$\RandomColoring$
takes $2$ rounds (rounds $3$ and $4$ of some phase).
Only an $\Active$ node may run the procedure, and when the procedure is over,
the node either stays in mode $\Active$ or moves to mode $\Colored$.
Consider some node $v$ running the procedure and let $C(v) \subseteq \{ 1, 2, 3
\}$ be the subset of colors which are not yet assigned to the neighbors of $v$
in $T$.
(Our analysis shows that if $v$ is $\Active$, then $C(v) \neq \emptyset$.)
As every $\Colored$ node transmits a message indicating its color, $v$ can
determine $C(v)$ by querying its ports.

In the first round of Procedure~$\RandomColoring$, $v$ picks some color $c \in
C(v)$ uniformly at random and transmits a `proposing color $c$' message.
In the second round of the procedure, if $v$ finds a `proposing color $c$'
(with the same $c$) in its ports, then it remains in mode $\Active$.
Otherwise (no neighbor of $v$ competes with $v$ over color $c$), it moves to
mode $\Colored$ and transmits a `my color is $c$' message.
This completes the description of our protocol.

%%%%%%%%%%%%%%%%%%%%%%%%%%%%%%%%%%%%%%%
\paragraph{The Waiting Hierarchy.}
%%%%%%%%%%%%%%%%%%%%%%%%%%%%%%%%%%%%%%%
The `waits on' relation induces a hierarchy referred to as the \emph{waiting
hierarchy} which is represented by a (collection of) directed tree(s) defined
over a subset of the edges of the tree $T$.
Our protocol is designed so that if $v$ waits on $u$, moving to mode
$\Waiting$ in phase $i$, then in phases $1, \dots, i$, $u$ was $\Active$, and in
phase $i + 1$, $u$ is either $\Active$ or $\Colored$.
Moreover, if $u$ is $\Active$ and $v \in \Neighbors(u)$ is $\Waiting$, then
$v$ must be waiting on $u$.
Note also that if $v$ waits on $u$ and $u$ moves to mode $\Colored$ in phase
$j$, then $v$ moves back to mode $\Active$ in (the beginning of) phase $j + 1$
and $d^{j + 1}(v) = 0$.

\begin{observation*}
In the beginning of phase $i$, $|C(v)| \geq \min \{ d^{i}(v) + 1, 3 \}$ for
every $i \in \Integers_{> 0}$ and node $v \in V^{i}$.
\end{observation*}
\begin{proof}
As long as $d^{i}(v) \geq 3$, no neighbor of $v$ can run
Procedure~$\RandomColoring$, and hence no neighbor of $v$ can move to mode
$\Colored$.
Therefore, $C(v) = \{ 1, 2, 3 \}$ in the beginning of the first phase $i \in
\Integers_{> 0}$ such that $d^{i}(v) \leq 2$.
From that moment on, every $\Active$ neighbor of $v$ that moves to mode
$\Colored$ decreases both $|C(v)|$ and $d^{i}(v)$ by $1$.
The assertion is completed by recalling that non-$\Active$ neighbors of $v$
must be waiting on $v$ and hence, cannot move to mode $\Colored$ before $v$
does.
\end{proof}

\begin{corollary} \label{corollary:GettingColored}
Consider some node $v \in V^{i}$ that runs Procedure~$\RandomColoring$.
If $d^{i}(v) = 0$, then $v$ moves to mode $\Colored$ with probability $1$.
Otherwise ($d^{i}(v)$ is either $1$ or $2$), $v$ moves to mode $\Colored$ with
a positive constant probability.
\end{corollary}

Let $\widetilde{V}^{i}$ be the restriction of $V^{i}$ to nodes $v$ that were
$\Active$ in all phases $1, \dots, i$;
this is, $\widetilde{V}^{i}$ does not include $\Waiting$ nodes that became
$\Active$ again (recall that these will move to mode $\Colored$ in the next
phase with probability $1$).
Let $\widetilde{F}^{i}$ be the forest induced on $T$ by $\widetilde{V}^{i}$.
Given some node $v \in \widetilde{V}^{i}$, let $\widetilde{\Neighbors}^{i}(v)
= \{ u \in \widetilde{V}^{i} \mid (u, v) \in E \}$ be the neighborhood of $v$
in $\widetilde{F}^{i}$ and let $\widetilde{d}^{i}(v) =
|\widetilde{\Neighbors}^{i}(v)|$ be its degree.
Observe that if $v \in \widetilde{V}^{i}$, then $v \in V^{i}$ and
$\widetilde{d}^{i}(v) = d^{i}(v)$.
Therefore, if $v \in V^{i} - \widetilde{V}^{i}$, then $d^{i}(v) = 0$, in which
case $v$ runs Procedure~$\RandomColoring$ in phase $i$ and
Corollary~\ref{corollary:GettingColored} guarantees that $v \notin V^{i + 1}$.

The correctness of the protocol can now be established:
The logic of Procedure~$\RandomColoring$ implies that every output
configuration is a legal coloring.
Since $\Active$ leaves are removed from $\widetilde{F}^{i}$ with probability
$1$ and since every tree has at least two leaves, it follows that
$\widehat{V}^{1 + k} = \emptyset$ for $k = \lceil n / 2 \rceil$.
Combining the properties of the waiting hierarchy with
Corollary~\ref{corollary:GettingColored}, we conclude that the execution
reaches an output configuration within at most $k$ additional phases.
It remains to analyze the run-time of our protocol.

%%%%%%%%%%%%%%%%%%%%%%%%%%%%%%%%%%%%%%%
\paragraph{Good nodes.}
%%%%%%%%%%%%%%%%%%%%%%%%%%%%%%%%%%%%%%%
Consider some tree $T'$.
We say that node $v$ of $T'$ is \emph{good} if $v$ is a leaf or if the
degree of $v$ is $2$ and both neighbors of $v$ are of degree at most $2$.

\begin{observation} \label{obervation:ManyGoodNodes}
In every tree, at least a $(1 / 5)$-fraction of the nodes are good. 
\end{observation}

Consider some $i \in \Integers_{> 0}$ and some node $v \in \widetilde{V}^{i}$.
Let $T'$ be the tree to which $v$ belongs in $\widetilde{F}^{i}$.
We argue that if $v$ is good in $T'$, then $v \notin \widetilde{V}^{i}$ with a
positive constant probability.
Indeed, if $v$ is a leaf in $T'$, which means that $\widetilde{d}^{i}(v) =
d^{i}(v) = 1$, then it either moves to mode $\Waiting$ with probability $1$
(if the neighbor of $v$ has a higher degree) or it runs
Procedure~$\RandomColoring$, in which case
Corollary~\ref{corollary:GettingColored} guarantees that $v$ moves to mode
$\Colored$ with a positive constant probability;
if $\widetilde{d}^{i}(v) = d^{i}(v) = 2$ and both neighbors of $v$ in
$F^{i}$ (and in $T'$) are of degree at most $2$, then $v$ runs
Procedure~$\RandomColoring$, in which case
Corollary~\ref{corollary:GettingColored} again guarantees that $v$ moves to
mode $\Colored$ with a positive constant probability.
Since Corollary~\ref{corollary:GettingColored} also guarantees that
nodes of degree $0$ in $\widetilde{F}^{i}$ move to mode $\Colored$ with
probability $1$, we can employ Observation~\ref{obervation:ManyGoodNodes}
and Markov's bound to establish the following observation.

\begin{observation} \label{observation:ManyActiveNodesVanish}
There exists two constants $0 < p, c < 1$ such that $|\widetilde{V}^{i + 1}|
\leq c |\widetilde{V}^{i}|$ with probability at least $p$.
\end{observation}

Similarly to the analysis in \Section{}~\ref{section:MIS}, define the random
variable $Y = \min\{ i \in \Integers_{> 0} : |\widetilde{V}^{i}| = 0 \}$.
Observation~\ref{observation:ManyActiveNodesVanish} implies
that $Y$ is stochastically dominated by a random variable that obeys
distribution $\NegativeBinomial(O (\log n), 1 - p) + O (\log n)$, namely, a
fixed term of $O (\log n)$ plus the negative binomial distribution with
parameters $O (\log n)$ and $1 - p$, hence $Y = O (\log n)$ in expectation and
with high probability.
Since $Y$ bounds from above the depth of the waiting hierarchy, it follows
that the execution reaches an output configuration within $2 Y$ phases, which
completes the analysis.

\begin{theorem} \label{theorem:TreeColoring}
There exists an nFSM protocol that $3$-colors any $n$-node (undirected) tree
with run-time $O (\log n)$.
\end{theorem}
} % end \SectionTreeColoring
\LongVersion %{
\SectionTreeColoring{}
\LongVersionEnd %}

\ShortVersion %{
\sloppy\ignorespaces
\ShortVersionEnd %}
%%%%%%%%%%%%%%%%%%%%%%%%%%%%%%%%%%%%%%%%%%%%%%%%%%%%%%%%%%%%%%%%%%%%%%%%%%%%%%
\section{Computational Power}
\label{section:LBA}
%%%%%%%%%%%%%%%%%%%%%%%%%%%%%%%%%%%%%%%%%%%%%%%%%%%%%%%%%%%%%%%%%%%%%%%%%%%%%%
A \emph{deterministic linear bounded automaton (dLBA)} is a (deterministic)
Turing machine whose working tape is restricted to the cells specifying the
input (this is equivalent to a $\mathrm{DSPACE}(O (n))$ Turing machine).
A \emph{non-deterministic linear bounded automaton}, a.k.a., \emph{linear
bounded automaton (LBA)}, is the non-deterministic version of a dLBA, and a
\emph{randomized linear bounded automaton (rLBA)} is the randomized version.
Kuroda~\cite{Kuro64} proved that the class of languages that can be decided by
an LBA is exactly the context-sensitive languages, corresponding to the Type-1
grammars in Chomsky's hierarchy of formal languages \cite{Chom56}.
Whether LBAs are equivalent to dLBAs and where exactly do rLBAs lie between
the two are major open questions in computational complexity (cf. the first
LBA problem).
The following two lemmas
\ShortVersion %{
(proofs are deferred to Appendices \ref{appendix:SimulateNfsm} and
\ref{appendix:SimulateRlba})
\ShortVersionEnd %}
show that in terms of its computational power (regardless of run-time
considerations), an nFSM protocol is essentially equivalent to an rLBA.
\ShortVersion %{
\par\fussy\ignorespacesafterend
\ShortVersionEnd %}

\begin{lemma} \label{lemma:SimulateNfsm}
An nFSM protocol on a graph $G$ of arbitrary topology can be simulated by an
rLBA.
\end{lemma}
\def\ProofLemmaSimulateNfsm{
The input for the Turing machine is the graph $G$, given as an adjacency
list.
In order to simulate the execution of the nFSM protocol, we store some
additional information in the entries of the adjacency list as follows:
For each node $v$, we store its current state and the next letter it
transmits.
For every node $u$ in the list of neighbors $N(v)$ attached to $v$, we store
the entry of $u$'s port that corresponds to $v$.
In each round of the nFSM protocol, the rLBA performs two sweeps of the list
of nodes:
The first sweep serves to calculate $v$'s next state $q$ and transmitted
letter $\sigma$ for all nodes $v$, based on $v$'s current state and the
messages in its ports, according to the nFSM state machine, which is
hard-wired in the rLBA.
However, the calculated letter $\sigma$ is not being ``transmitted'' yet, so
the calculations for subsequent nodes in the list are not messed up, but
rather stored in the corresponding place next to $v$.
In the second sweep, for every node $v$, the letter $\sigma$ is being
``transmitted'', that is, the lists of neighbors are traversed, and at each
occurrence of $v$, the current letter is replaced by $\sigma$.
This way, we simulate every round of the nFSM protocol.
In total, our simulation requires additional $O(1)$ space per node and $O(1)$
space per edge, hence it can be implemented with an rLBA.
The assertion follows.
} % end \ProofLemmaSimulateNfsm
\LongVersion %{
\begin{proof}
\ProofLemmaSimulateNfsm{}
\end{proof}
\LongVersionEnd %}

\begin{lemma} \label{lemma:SimulateRlba}
An rLBA can be simulated by an nFSM protocol on a path.
\end{lemma}
\def\ProofLemmaSimulateRlba{
Let $n$ be the number of cells in the tape of the rLBA.
Then, the path network has $n$ nodes, each corresponding to one cell of the
tape, i.e., we identify a node $v$ of the path nFSM with a certain cell on the
tape.
Let $\Gamma$ be the working alphabet and $P$ be the state space of the rLBA.
The nFSM protocol is designed so that the state of node $v$ indicates:
(1) which letter from $\Gamma$ is written in $v$;
(2) if the head of the rLBA currently points to $v$;
(3) the current state of the rLBA, which is allowed to be incorrect if (2) is
false; and
(4) if the head is currently located to the left or to the right of $v$.
Hence, we fix $Q = \Gamma \times \{ 0, 1 \} \times P \times  \{ L, R \}$.
The alphabet of the nFSM is $\Sigma = \{ L, R \} \times P$.

Suppose that the input to the rLBA is $\gamma_1 \dots \gamma_n \in \Gamma^n$.
Then, we assume that the initial state of the $i$th node in the path is
$(\gamma_i, h, p_0, L)$, where $p_0$ is the initial state of the Turing
machine and
\[
h =
\begin{cases}
1 \quad\text{if } i = 1 \\
0 \quad\text{if } i > 1 \, .
\end{cases}
\]
Note that the distinction between the initial state of the first node in the
path and the initial states of all other nodes is without loss of generality.
Indeed, as the first and last nodes have degree $1$ and all interior nodes
have degree $2$, it is easy for a node to ``decide'' (under the nFSM model) if
it is an interior node.
Distinguishing between the first and last nodes is unavoidable if one wants to
distinguish between the inputs $\gamma_1\dots \gamma_n$ and $\gamma_n \dots
\gamma_1$.

At all times, we maintain the invariant that exactly one node is in a state in
$\Gamma \times \{1\} \times P \times \{ L, R \}$ --- denote this node as
\emph{active} --- whereas all other nodes are in a state in $\Gamma \times
\{ 0 \} \times P \times \{ L, R \}$.
Only the active node can transmits messages;
all other nodes remain silent and listen.
If an non-active node $v$ receives a message indicating that the head should
move to the left (respectively, right), and $v$'s state indicates that the
head is currently to its right (resp., left), then $v$ becomes the active
node;
otherwise, $v$ does not react to this message.
Now, the nodes simulate the behavior of the rLBA by calculating the next state
of the rLBA based on the rLBA's transition function (which is hard-wired in the
FSM) and updating their own states accordingly.
The assertion follows.
} % end \ProofLemmaSimulateRlba
\LongVersion %{
\begin{proof}
\ProofLemmaSimulateRlba{}
\end{proof}
\LongVersionEnd %}

%%%%%%%%%%%%%%%%%%%%%%%%%%%%%%%%%%%%%%%%%%%%%%%%%%%%%%%%%%%%%%%%%%%%%%%%%%%%%%
%%%%%%%%%%%%%%%%%%%%%%%%%%%%%%%%%%%%%%%%%%%%%%%%%%%%%%%%%%%%%%%%%%%%%%%%%%%%%%
\ShortVersion %{
\clearpage

\pagenumbering{roman}
\appendix

\renewcommand{\theequation}{A-\arabic{equation}}
\setcounter{equation}{0}

\begin{center}
\textbf{\Large{APPENDIX}}
\end{center}

%%%%%%%%%%%%%%%%%%%%%%%%%%%%%%%%%%%%%%%%%%%%%%%%%%%%%%%%%%%%%%%%%%%%%%%%%%%%%%
\section{Proving Theorem~\ref{theorem:Synchronizer}}
\label{appendix:Synchronizer}
%%%%%%%%%%%%%%%%%%%%%%%%%%%%%%%%%%%%%%%%%%%%%%%%%%%%%%%%%%%%%%%%%%%%%%%%%%%%%%
Our goal in this section is to design a synchronizer for the nFSM model, thus
establishing Theorem~\ref{theorem:Synchronizer}.

\DetailsSynchronizer{}

%%%%%%%%%%%%%%%%%%%%%%%%%%%%%%%%%%%%%%%%%%%%%%%%%%%%%%%%%%%%%%%%%%%%%%%%%%%%%%
\section{Coloring a Tree with $3$ Colors}
\label{appendix:Coloring}
%%%%%%%%%%%%%%%%%%%%%%%%%%%%%%%%%%%%%%%%%%%%%%%%%%%%%%%%%%%%%%%%%%%%%%%%%%%%%%
\SectionTreeColoring{}

%%%%%%%%%%%%%%%%%%%%%%%%%%%%%%%%%%%%%%%%%%%%%%%%%%%%%%%%%%%%%%%%%%%%%%%%%%%%%%
\section{Proving Lemma~\ref{lemma:SimulateNfsm}}
\label{appendix:SimulateNfsm}
%%%%%%%%%%%%%%%%%%%%%%%%%%%%%%%%%%%%%%%%%%%%%%%%%%%%%%%%%%%%%%%%%%%%%%%%%%%%%%
\ProofLemmaSimulateNfsm{}

%%%%%%%%%%%%%%%%%%%%%%%%%%%%%%%%%%%%%%%%%%%%%%%%%%%%%%%%%%%%%%%%%%%%%%%%%%%%%%
\section{Proving Lemma~\ref{lemma:SimulateRlba}}
\label{appendix:SimulateRlba}
%%%%%%%%%%%%%%%%%%%%%%%%%%%%%%%%%%%%%%%%%%%%%%%%%%%%%%%%%%%%%%%%%%%%%%%%%%%%%%
\ProofLemmaSimulateRlba{}

\ShortVersionEnd %}
%%%%%%%%%%%%%%%%%%%%%%%%%%%%%%%%%%%%%%%%%%%%%%%%%%%%%%%%%%%%%%%%%%%%%%%%%%%%%%
%%%%%%%%%%%%%%%%%%%%%%%%%%%%%%%%%%%%%%%%%%%%%%%%%%%%%%%%%%%%%%%%%%%%%%%%%%%%%%
\clearpage
\renewcommand{\thepage}{}

\bibliographystyle{abbrv}
\bibliography{references}

\begin{thebibliography}{10}

\bibitem{AfekAlonBarJoseph+11}
Y.~Afek, N.~Alon, Z.~Bar-Joseph, A.~Cornejo, B.~Haeupler, and F.~Kuhn.
\newblock Beeping a maximal independent set.
\newblock In {\em Proceedings of the 25th international conference on
  Distributed computing (DISC)}, pages 32--50, 2011.

\bibitem{AfekAlonBarad+11}
Y.~Afek, N.~Alon, O.~Barad, E.~Hornstein, N.~Barkai, and Z.~Bar-Joseph.
\newblock {A Biological Solution to a Fundamental Distributed Computing
  Problem}.
\newblock {\em Science}, 331(6014):183--185, Jan. 2011.

\bibitem{ABI86}
N.~Alon, L.~Babai, and A.~Itai.
\newblock A fast and simple randomized parallel algorithm for the maximal
  independent set problem.
\newblock {\em J. Algorithms}, 7:567--583, December 1986.

\bibitem{Awer85}
B.~Awerbuch.
\newblock Complexity of network synchronization.
\newblock {\em J. ACM}, 32(4):804--823, 1985.

\bibitem{APPS92}
B.~Awerbuch, B.~Patt-Shamir, D.~Peleg, and M.~E. Saks.
\newblock Adapting to asynchronous dynamic networks (extended abstract).
\newblock In {\em STOC}, pages 557--570, 1992.

\bibitem{AP90}
B.~Awerbuch and D.~Peleg.
\newblock Network synchronization with polylogarithmic overhead.
\newblock In {\em FOCS}, pages 514--522, 1990.

\bibitem{BE09}
L.~Barenboim and M.~Elkin.
\newblock Distributed (delta+1)-coloring in linear (in delta) time.
\newblock In {\em STOC}, pages 111--120, 2009.

\bibitem{BE11b}
L.~Barenboim and M.~Elkin.
\newblock Combinatorial algorithms for distributed graph coloring.
\newblock In {\em DISC}, pages 66--81, 2011.

\bibitem{BE11a}
L.~Barenboim and M.~Elkin.
\newblock Deterministic distributed vertex coloring in polylogarithmic time.
\newblock {\em J. ACM}, 58(5):23, 2011.

\bibitem{BellareGS98}
M.~Bellare, O.~Goldreich, and M.~Sudan.
\newblock Free bits, pcps, and nonapproximability-towards tight results.
\newblock {\em SIAM J. Comput.}, 27(3):804--915, 1998.

\bibitem{Benenson+01}
Y.~Benenson, T.~Paz-Elizur, R.~Adar, E.~Keinan, Z.~Livneh, and E.~Shapiro.
\newblock {Programmable and autonomous computing machine made of biomolecules}.
\newblock {\em Nature}, 414(6862):430--434, Nov. 2001.

\bibitem{BrandZafiropulo83}
D.~Brand and P.~Zafiropulo.
\newblock On communicating finite-state machines.
\newblock {\em J. ACM}, 30:323--342, April 1983.

\bibitem{ChlamtacKuten85}
I.~Chlamtac and S.~Kutten.
\newblock {On Broadcasting in Radio Networks--Problem Analysis and Protocol
  Design}.
\newblock {\em Communications, IEEE Transactions on [legacy, pre - 1988]},
  33(12):1240--1246, 1985.

\bibitem{Chom56}
N.~Chomsky.
\newblock {Three models for the description of language}.
\newblock {\em IRE Transactions on Information Theory}, 2:113--124, 1956.
\newblock \url{http://www.chomsky.info/articles/195609--.pdf}.

\bibitem{CV86}
R.~Cole and U.~Vishkin.
\newblock Deterministic coin tossing with applications to optimal parallel list
  ranking.
\newblock {\em Inf. Control}, 70(1):32--53, July 1986.

\bibitem{CornejoKuhn10}
A.~Cornejo and F.~Kuhn.
\newblock Deploying wireless networks with beeps.
\newblock In {\em Proceedings of the 24th international conference on
  Distributed computing (DISC)}, pages 148--162, 2010.

\bibitem{Flury2010Slotted}
R.~Flury and R.~Wattenhofer.
\newblock {Slotted Programming for Sensor Networks}.
\newblock In {\em {International Conference on Information Processing in Sensor
  Networks (IPSN), Stockholm, Sweden}}, April 2010.

\bibitem{Gardner70}
M.~Gardner.
\newblock {The fantastic combinations of John Conway's new solitaire game
  `life'}.
\newblock {\em Scientific American}, 223(4):120--123, 1970.

\bibitem{GPS88}
A.~V. Goldberg, S.~A. Plotkin, and G.~E. Shannon.
\newblock Parallel symmetry-breaking in sparse graphs.
\newblock {\em SIAM J. Discrete Math.}, 1(4):434--446, 1988.

\bibitem{Gordon04}
P.~Gordon.
\newblock {Numerical Cognition Without Words: Evidence from Amazonia}.
\newblock {\em Science}, 306(5695):496--499, Oct. 2004.

\bibitem{KothapalliSOS06}
K.~Kothapalli, C.~Scheideler, M.~Onus, and C.~Schindelhauer.
\newblock {Distributed Coloring in $\tilde O(\sqrt{\log n})$ Bit Rounds}.
\newblock In {\em 20th International Parallel and Distributed Processing
  Symposium (IPDPS)}, 2006.

\bibitem{Kuhn09}
F.~Kuhn.
\newblock Weak graph colorings: distributed algorithms and applications.
\newblock In {\em Proceedings of the twenty-first annual symposium on
  Parallelism in algorithms and architectures}, SPAA '09, pages 138--144, New
  York, NY, USA, 2009. ACM.

\bibitem{KMW04}
F.~Kuhn, T.~Moscibroda, and R.~Wattenhofer.
\newblock What cannot be computed locally!
\newblock In {\em Proceedings of the twenty-third annual ACM symposium on
  Principles of distributed computing (PODC)}, pages 300--309, 2004.

\bibitem{Kuro64}
S.-Y. Kuroda.
\newblock Classes of languages and linear-bounded automata.
\newblock {\em Information and Control}, 7(2):207--223, 1964.

\bibitem{LW11}
C.~Lenzen and R.~Wattenhofer.
\newblock {MIS on trees}.
\newblock In {\em Proceedings of the 30th annual ACM SIGACT-SIGOPS symposium on
  Principles of distributed computing (PODC)}, pages 41--48, New York, NY, USA,
  2011.

\bibitem{Linial92}
N.~Linial.
\newblock Locality in distributed graph algorithms.
\newblock {\em SIAM J. Comput.}, 21:193--201, Feb. 1992.

\bibitem{Luby86}
M.~Luby.
\newblock A simple parallel algorithm for the maximal independent set problem.
\newblock {\em SIAM J. Comput.}, 15:1036--1055, November 1986.

\bibitem{LynchBook}
N.~A. Lynch.
\newblock {\em Distributed Algorithms}.
\newblock Morgan Kaufmann, 1st edition, 1996.

\bibitem{MRSZ11}
Y.~M{\'e}tivier, J.~M. Robson, N.~Saheb-Djahromi, and A.~Zemmari.
\newblock {An optimal bit complexity randomised distributed MIS algorithm}.
\newblock {\em Distributed Computing}, 23(5-6):331--340, Jan. 2011.

\bibitem{vonNeumann66}
J.~V. Neumann.
\newblock {\em Theory of Self-Reproducing Automata}.
\newblock University of Illinois Press, Champaign, IL, USA, 1966.

\bibitem{PelegBook}
D.~Peleg.
\newblock {\em Distributed computing: a locality-sensitive approach}.
\newblock Society for Industrial and Applied Mathematics, Philadelphia, PA,
  USA, 2000.

\bibitem{Plot88}
S.~Plotkin.
\newblock {\em Graph-theoretic techniques for parallel, distributed, and
  sequential computation}.
\newblock MIT/LCS/TR. Laboratory for Computer Science, Massachusetts Institute
  of Technology, 1988.

\bibitem{BiologyTextBook}
D.~Sadava.
\newblock {\em Life: The Science of Biology}.
\newblock Sinauer Associates, 2011.

\bibitem{SW10}
J.~Schneider and R.~Wattenhofer.
\newblock {An Optimal Maximal Independent Set Algorithm for
  Bounded-Independence Graphs}.
\newblock In {\em {Journal of Distributed Computing}}, March 2010.

\bibitem{Schneider2011ChromaticNumber}
J.~Schneider and R.~Wattenhofer.
\newblock {Distributed Coloring Depending on the Chromatic Number or the
  Neighborhood Growth}.
\newblock In {\em {18th International Colloquium on Structural Information and
  Communication Complexity (SIROCCO), Poland}}, June 2011.

\bibitem{JukkaSurvey}
J.~Suomela.
\newblock Survey of local algorithms.
\newblock {\em To appear in: ACM Computing Surveys}, 2012.
\newblock \url{http://www.cs.helsinki.fi/u/josuomel/doc/local-survey.pdf}.

\bibitem{SV93}
M.~Szegedy and S.~Vishwanathan.
\newblock Locality based graph coloring.
\newblock In {\em STOC}, pages 201--207, 1993.

\bibitem{Wolfram02}
S.~Wolfram.
\newblock {\em {A new kind of science}}.
\newblock Wolfram Media, Champaign, Illinois, 2002.

\end{thebibliography}

\end{document}